\documentclass[12pt]{amsart}
\usepackage{amssymb}
\usepackage{hyperref}

\usepackage{tikz}\usetikzlibrary{calc,decorations.pathreplacing,decorations.pathmorphing,positioning}
\let\frak\mathfrak
\let\Bbb\mathbb

\usepackage{eucal}

\def\>{\relax\ifmmode\mskip.666667\thinmuskip\relax\else\kern.111111em\fi}
\def\<{\relax\ifmmode\mskip-.333333\thinmuskip\relax\else\kern-.0555556em\fi}
\def\vsk#1>{\vskip#1\baselineskip}
\def\vv#1>{\vadjust{\vsk#1>}\ignorespaces}
\def\vvn#1>{\vadjust{\nobreak\vsk#1>\nobreak}\ignorespaces}
 \let\alb\allowbreak

\def\plait#1{\par\hangindent2\parindent\indent\kern\parindent
\llap{#1\enspace}\ignorespaces}

\let\Smallskip\smallskip
\def\smallskip{\par\Smallskip}
\let\Medskip\medskip
\def\medskip{\par\Medskip}
\let\Bigskip\bigskip
\def\bigskip{\par\Bigskip}

\let\Maketitle\maketitle
\def\maketitle{\Maketitle\thispagestyle{empty}\let\maketitle\empty}

\newtheorem{thm}{Theorem}[section]
\newtheorem*{thm*}{Theorem}
\newtheorem{cor}[thm]{Corollary}
\newtheorem{lem}[thm]{Lemma}
\newtheorem{prop}[thm]{Proposition}

\newtheorem{defn}[thm]{Definition}

\numberwithin{equation}{section}

\theoremstyle{definition}
\newtheorem*{rem}{Remark}
\newtheorem*{example}{Example}

\def\beq{\begin{equation}}
\def\eeq{\end{equation}}
\def\be{\begin{equation*}}
\def\ee{\end{equation*}}

\def\bean{\begin{eqnarray}}
\def\eean{\end{eqnarray}}
\def\bea{\begin{eqnarray*}}
\def\eea{\end{eqnarray*}}

\let\al\alpha

\let\dl\delta  
 \let\eps\varepsilon \let\epsilon\eps

\let\la\lambda 

\let\si\sigma

\newcommand{\glt}{{\frak{gl}_2}}
\newcommand{\Il}{{\mc L_{\la}}}

\let\der\partial

\let\ge\geqslant
\let\geq\geqslant
\let\le\leqslant
\let\leq\leqslant

\let\ox\otimes

\def\R{\Bbb R}
\def\C{\Bbb C}
\def\N{\Bbb N}

\def\gl{\frak{gl}}

\def\lsym#1{#1\alb\dots\relax#1\alb} \def\lc{\lsym,}

\let\on\operatorname

\def\End{\on{End}}

\def\gln{\gl_N}

\def\Yn{Y(\gln)}

\let\mc\mathcal

\def\ii{i,\<\>i}
\def\ij{i,\<\>j}
\def\ji{j,\<\>i}
\def\js{j,\<\>s}

\def\ik{i,\<\>k}
\def\il{i,\<\>l}
\def\jk{j,\<\>k}
\def\kj{k,\<\>j}
\def\kl{k,\<\>l}
\def\sj{s,\<\>j}
\def\sk{s,\<\>k}

\def\KZ/{KZ}
\def\qKZ/{{\itshape q}\/KZ}
\def\qi-{{\itshape q\/}-}

\DeclareMathOperator{\sgn}{sgn}
\def\lrar{\leftrightarrow}

\allowdisplaybreaks[1]

\title[Extended Joseph polynomials, \qi-conformal blocks, and a \qi-Selberg
integral]{Extended Joseph polynomials, quantized conformal blocks,
and a \qi-Selberg type integral}

\author[R.\,Rim\'anyi]{R.\ Rim\'anyi$\>^{1}$}
\author[V.\,Tarasov]{V.\ Tarasov$\>^2$}
\author[A.\,Varchenko]{A.\ Varchenko$\>^{3}$}
\author[P.\,Zinn-Justin]{P.\ Zinn-Justin$\>^{4}$}

\address{R.~Rim\'anyi, Department of Mathematics, University of North Carolina
at Chapel Hill, Chapel Hill, NC 27599-3250, USA.}

\address{V.~Tarasov,
Department of Mathematical Sciences,
Indiana University\,--\>Purdue University Indianapolis,
402 North Blackford St, Indianapolis, IN 46202-3216, USA
and
St.~Petersburg Branch of Steklov Mathematical Institute
Fontanka 27, St.\,Petersburg, 191023, Russia.}

\address{A.~Varchenko, Department of Mathematics, University of North Carolina
at Chapel Hill, Chapel Hill, NC 27599-3250, USA.}

\address{P.~Zinn-Justin,
UPMC Univ Paris 6, CNRS UMR 7589, LPTHE,
75252 Paris Cedex, France.}

\begin{document}
\maketitle

{\hfill\it To the memory of Yu.~Stroganov}

{\let\thefootnote\relax
\footnotetext{\vsk-.8>\noindent
$^1\<${\sl E\>-mail}:\enspace rimanyi@email.unc.edu\>,
supported in part by NSA grant CON:H98230-10-1-0171.\\
$^2\<${\sl E\>-mail}:\enspace vt@math.iupui.edu\>, vt@pdmi.ras.ru\>,
supported in part by NSF grant DMS-0901616.\\
$^3\<${\sl E\>-mail}:\enspace anv@email.unc.edu\>,
supported in part by NSF grant DMS-1101508.\\
$^4\<${\sl E\>-mail}:\enspace pzinn@lpthe.jussieu.fr,
supported in part by ERC grant 278124 ``LIC''.}}

\bigskip

\medskip
\begin{abstract}

We consider the tensor product
$V=(\C^N)^{\otimes n}$ of the vector representation of $\gln$
and its weight decomposition $V=\oplus_{\la=(\la_1,\dots,\la_N)}V[\la]$.
For $\la = (\la_1\geq \dots \geq \la_N)$,
the trivial bundle $V[\la]\times \C^n\to\C^n$
has a subbundle
of \qi-conformal blocks at
level $\ell$, where $\ell = \la_1-\la_N$
if $\la_1-\la_N> 0$ and $\ell=1$ if $\la_1-\la_N=0$.
We construct a polynomial section $I_\la(z_1,\dots,z_n,h)$ of the
subbundle. The section  is the main object of the paper. We identify
the section with
 the generating function $J_\la(z_1,\dots,z_n,h)$ of the extended Joseph
polynomials of orbital varieties, defined in \cite{artic34, {artic39}}.

For $\ell=1$, we show that the subbundle of \qi-conformal blocks has rank 1
and $I_\la(z_1,\dots,z_n,h)$ is  flat with respect to the quantum Knizhnik--Zamolodchikov discrete connection.

For $N=2$ and $\ell=1$, we represent our polynomial as a multidimensional
\qi-hypergeometric integral and obtain a \qi-Selberg type
identity, which says that the integral is an explicit polynomial.
\end{abstract}

\setcounter{footnote}{0}
\renewcommand{\thefootnote}{\arabic{footnote}}

\section{Introduction}
The bundle of conformal blocks was introduced in conformal field theory.
The bundle has a projectively flat Knizhnik--Zamolodchikov (KZ) connection which is a flat connection for
conformal blocks on the sphere, see, for example, \cite{KZ, KL-tensor}.
The equations for flat sections of the bundle of conformal blocks on the sphere
are called the KZ differential equations. In \cite{SV, FScV, FScV2} solutions of the
KZ differential equations were constructed as multidimensional hypergeometric integrals.

The \qKZ/ difference equations were introduced in \cite{FR-qKZ}. The
$\gln$ \qi-conformal blocks were defined in \cite{MV1, MV2}, cf~\cite{EF}.
The  bundle of \qi-conformal blocks has a discrete flat connection
defined by \qKZ/ operators, see  \cite{MV1, MV2}. In \cite{TV1} solutions of the $\glt$ \qKZ/
equations were constructed as multidimensional \qi-hypergeometric integrals.

We consider the tensor product 
$V=(\C^N)^{\otimes n}$ of the vector representation of $\gln$
and its weight decomposition $V=\oplus_{\la=(\la_1,\dots,\la_N)}V[\la]$.
For $\la = (\la_1\geq \dots \geq \la_N)$,
the trivial bundle $V[\la]\times \C^n\to\C^n$
has a subbundle
of \qi-conformal blocks at
level $\ell$, where $\ell = \la_1-\la_N$
if $\la_1-\la_N> 0$ and $\ell=1$ if $\la_1-\la_N=0$.
We construct a polynomial section $I_\la(z_1,\dots,z_n,h)$ of the
subbundle. The section  is the main object of the paper. We identify
the section with
 the generating function $J_\la(z_1,\dots,z_n,h)$ of the extended Joseph
polynomials of the orbital varieties, defined in \cite{artic34, {artic39}}.

For $\ell=1$, we show that the subbundle of \qi-conformal blocks has rank 1
and $I_\la(z_1,\dots,z_n,h)$ is  flat with respect to the \qKZ/ discrete connection.

For $N=2$ and $\ell=1$, we represent our polynomial as a multidimensional
\qi-hypergeometric integral and obtain a \qi-Selberg type
identity, which says that the integral is an explicit polynomial.
The simplest of these identities is
\beq
\label{barnes}
\int_{-i\infty}^{i\infty}
\Gamma(a + s)\Gamma(b + s)\Gamma(c - s)\Gamma(d - s) \,ds \,=
\,2\pi i\ \frac{\Gamma(a + c)\Gamma(a + d)\Gamma(b + c)\Gamma(b + d)}
{\Gamma(a + b + c + d)}\;,\kern-.5em
\eeq
which is a formula for the Barnes integral in \cite{WW}.
The integral representation for $I_\la$ gives an integral representation for
the extended Joseph polynomials if $N=2$ and $\ell=1$.

For $N=2$ we also give a presentation for $I_\la$ as a multiple residue
of a suitable rational function.
\medskip

The fact that the generating function $J_\la(z_1,\dots,z_n,h)$ with $\la_1-\la_N\leq 1$ satisfies the \qKZ/
equations at level 1 was conjectured in \cite{artic34} and proved in \cite{KNST} by a different method, by
relating the generating function with non-symmetric Jack polynomials.

\medskip
The results of this paper may be considered as a ``quantization'' of the results of \cite{V, RV, RSV},
where the bundle of (non-quantum) conformal blocks at level 1
in $(\C^N)^{\otimes n}$ was considered.
The bundle is of rank 1 and has a flat connection
defined by the \KZ/ differential operators. 
A rational flat section of the bundle was constructed.
The section was interpreted as a generating function
of the Euler classes of the fixed points of the $GL_n$-action
on a suitable space of partial flags in $\C^n$. A Selberg type identity
was obtained that equates the rational section and
a multidimensional hypergeometric integral.

\medskip
In Section~\ref{sec:qkz} we introduce \qi-conformal blocks and \qKZ/ equations.
In Section~\ref{sec:Ila} we define our main object -- the polynomial
$I_\la(z_1,\dots,z_n,h)$. In Section~\ref{sec:geom} we identify the polynomial
with the generating function of the extended Joseph polynomials of orbital varieties.
In Section~\ref{sec:proofs} we prove all properties of the generating function.
In Section~\ref{sec:Selberg} we prove a \qi-Selberg type identity.
In Section~\ref{sec:altint} we give an alternative integral formula
for $I_\la(z_1,\dots,z_n,h)$, if $N=2$.





\section{Quantum conformal blocks and \qKZ/ equations}
\label{sec:qkz}

\subsection{Operators on representation-valued functions}
\label{sec:ops}

Let $N\ge 2$ be a positive integer. Let $e_{\ij}$, for $i,j=1,\dots,N$, be the
standard generators of the complex Lie algebra $\gln$ satisfying the relations
$[e_{\ij},e_{\sk}]=\delta_{\js}e_{\ik}-\delta_{\ik}e_{\sj}$.
Consider the standard vector
representation $\C^N$ of $\gln$, and its $n$-th tensor
product $V=(\C^N)^{\otimes n}$. The space $V$ splits into the
direct sum of weight subspaces
\be
V=\bigoplus_{\la=(\la_1,\ldots,\la_N)} V[\la],
\ee
where $\sum_{i=1}^N\la_i=n$, and $V[\la]=\{v\in V\ |\ e_{i,i}v=\la_i v\}.$

In this paper we will be concerned with $V$-valued functions of
$z_1,\ldots,z_n$ also depending on a complex parameter $h$.
Now we recall some operators acting on the space of such functions.

\begin{itemize}
\item{}
Elements of $\gln$ act naturally on any factor of the tensor product.
When $x\in\gln$ acts in the $i$-th factor, we denote its action by $x^{(i)}$.

\item{}
Following \cite{MV1} define an operator
\be
e(z)=\sum_{j=1}^n\,\Bigl(z_j-he_{N,N}^{(j)}+
h\sum_{s=j+1}^n\bigl(e_{1,1}^{(s)}-e_{N,N}^{(s)}\bigr)\Bigr)\,e_{1,N}^{(j)}
+ h\sum_{j=2}^{N-1}\sum_{1\le r< s\le N}\!e_{j,N}^{(r)}\,e_{1,j}^{(s)}.
\ee

\item{} Let $P^{(\ij)}$ be the permutation of the $i$-th and $j$-th factors
of $\left(\C^N\right)^{\otimes n}$.

\medskip
\item{}
{\it The deformed $S_n$-action on $V$-valued functions of $z_1\lc z_n$.}
The $i$-th elementary transposition $s_i\in S_n$ acts by the formula
\begin{align}
\label{Sn}
s_i: {}&I(z_1,\ldots,z_n)\,\mapsto{}
\\
& \frac{(z_i-z_{i+1})\,P^{(\ii+1)}+h}{z_i-z_{i+1}}\;
I(\ldots,z_{i+1},z_i,\ldots)-I(\ldots,z_i,z_{i+1},\ldots)\frac{h}{z_i-z_{i+1}}
\notag
\end{align}
This defines an action of $S_n$. Observe that, despite the presence of
denominators, polynomials are mapped to polynomials by elements of $S_n$. In
the whole paper an $S_n$-action will always mean this deformed action unless
otherwise stated.

\medskip
\item{} Let $u$ be a new variable. We define
the following $R$-matrix operator
\be
R^{(\ij)}(u)=\frac{u-h P^{(\ij)}}{u+h}.
\ee
Observe that $R^{(\ij)}(u)R^{(\ij)}(-u)=1$ \,and
\be
\label{YB}
R^{(\ij)}(u-v)\>R^{(\ik)}(u)\>R^{(\jk)}(v)\,=\,
R^{(\jk)}(v)\>R^{(\ik)}(u)\>R^{(\ij)}(u-v)\,.
\ee
\end{itemize}

\subsection{Yangian \,$\Yn$}

The Yangian $\Yn$ is a unital associative algebra with generators
$T_{\ij}^{\{s\}}$, $i,j=1\lc N$, \,$s\in\N$. Organize them into generating
series
\be
T_{\ij}(u)\,=\,\dl_{\ij}+\sum_{s=1}^\infty\,T_{\ij}^{\{s\}}\>u^{-s},\qquad
i,j=1\lc N\,.
\ee
The defining relation in $\Yn$ have the form
\beq
\label{ijkl}
(u-v)\,\bigl[T_{\ij}(u),T_{\kl}(v)\bigr]\,=\,
T_{\kj}(v)\>T_{\il}(u)-T_{\kj}(u)\>T_{\il}(v)\,,
\eeq
for all \,$i,j,k,l=1\lc N$.

The Yangian $\Yn$ contains $U(\gln)$ as a subalgebra. The embedding is given
by \;$e_{\ij}\mapsto T_{\ji}^{\{1\}}$ for any $i,j=1\lc N$.

Let
\be
T(u)\,=\,\sum_{\ij=1}^N\,E_{\ij}\ox T_{\ij}(u)
\ee
where $E_{\ij}$ is the image of \,$e_{\ij}\in\gln$ in $\End(\C^N)$.
Relations~\eqref{ijkl} can be written as the equality of series with
coefficients in \,$\End(\C^N\!\ox\C^N)\ox\Yn$\,:
\be
(u-v+P)\,T^{(1)}(u)\>T^{(2)}(v)\,=\,T^{(2)}(v)\>T^{(1)}(u)\,(u-v+P)\,,
\ee
where \,$P$ is the permutation of the \,$\C^N$ factors,
\;$T^{(1)}(u)=\sum_{\ij=1}^N E_{\ij}\ox 1\ox T_{\ij}(u)$
\;and \;$T^{(2)}(u)=1\ox T(u)$\,.

\smallskip
More information on the Yangian $\Yn$ can be found in~\cite{Mo}. Notice
that the series $T_{\ij}(u)$ here corresponds to the series $T_{j,i}(u)$
in~\cite{Mo}.

\smallskip
The assignment
\be
T(u)\,\mapsto\,R^{(0,1)}(z_1-hu)\dots R^{(0,n)}(z_n-hu)\;
\prod_{i=1}^n\,\frac{z_i-hu+h}{z_i-hu}
\ee
defines an action of the Yangian $\Yn$ on $V$-valued functions of $z_1\lc z_n$.
We identify here the space $(\C^N)^{\ox(n+1)}$ with $\C^N\ox V$ and count the
tensor factors by $0,1\lc n$.

The action of \,$e(z)$ on $V$-valued functions coincide with that of
\;$h\bigl(T_{N,1}^{\{2\}}-T_{N,N}^{\{1\}}\>T_{N,1}^{\{1\}}\bigr)$\,.


\begin{lem}
The Yangian action commutes with the $S_n$ action~\eqref{Sn}.
\end{lem}
\begin{proof}
The commutativity
with the first term in \eqref{Sn} follows from the Yang-Baxter equation for $R(u)$,
the last formula in Section \ref{sec:ops}. The commutativity with the second term in
\eqref{Sn} is the commutativity with multiplication by functions of $z_1,\dots z_n$.
\end{proof}

\subsection{Singular vectors, \qi-conformal blocks, and \qKZ/ equations}

Let $\la$ be a partition, i.e. assume that
$\la_1\ge\ldots\ge\la_N$.
Define $d(\la)=\la_1-\la_N$.

A vector $v\in V[\la]$ is a {\em singular vector},
if $\sum_{a=1}^N e^{(a)}_{\ij}v=0$ for all $i<j$.

Let $\ell\ge d(\la)$
be a positive integer, and $z=(z_1,\ldots,z_n)\in \C^n$. Following \cite{MV1} we call
$v\in V[\la]$ a {\em level $\ell$ \qi-conformal block}, if it is a singular vector and
\be
e(z)^{\ell-d(\la)+1}v=0.
\ee
Note that if $v\in V[\la]$ is a level $\ell$ \qi-conformal block, then $v$ is a level $\ell'$ \qi-conformal block
for any $\ell'>\ell$.

For $i=1,\ldots,n$, define the \qKZ/ operators at level 1 by the formula
\begin{multline*}
K_i(z_1,\dots,z_n)=R^{(\ii-1)}\left(z_i-z_{i-1}-(N+1)h\right)\cdots
R^{(i,1)}\left(z_i-z_1-(N+1)h\right)\times
\\
\times R^{(i,n)}(z_i-z_n)\cdots R^{(\ii+1)}(z_i-z_{i+1}).
\end{multline*}
The \qKZ/ difference equations at level 1 for a $V[\la]$-valued function $I$ is
the system of equations
\beq
\label{eq:qkz}
I(z_1,\ldots, z_i-(N+1)h,\ldots,z_n)=
K_i(z_1,\ldots,z_n)I(z_1,\ldots,z_i,\ldots,z_n), \ \quad i=1,\ldots,n.\kern-1em
\eeq

\begin{lem} \label{lem:dim}
Let $d(\la)\le 1$. For generic $z=(z_1,\ldots,z_n)\in\C^n$
the space of \qi-conformal blocks at level 1 is at most one-dimensional.
\end{lem}

\begin{proof}
The space of conformal blocks for $h=0$ is defined as
\be
CB_\la(z)=\bigl\{v\in V[\la] \ \text{is a singular vector and}\ \;
\Bigl(\,\sum_{j=1}^N z_ae^{(j)}_{1,N}\Bigr)^{\ell-d(\la)+1}v=0\,\bigr\}.
\ee
For generic $z\in \C^n$ the dimension of $CB_\la(z)$ is calculated by the Verlinde formula. For $\ell=1$ and $d(\la)\le 1$ the Verlinde
formula gives 1, so in this case the space of (non-quantum) conformal blocks for generic $z$ is one-dimensional.
The space of \qi-conformal blocks specializes to $CB_\la(z)$ at $h=0$. At this specialization the dimension may only increase, hence the
dimension of \qi-conformal blocks is at most 1.
\end{proof}

Below we will show that for generic $z$ the dimension is equal to 1.

\section{The minimal degree skew-symmetric polynomial $I_{\la}$}
\label{sec:Ila}

Recall that $\la\in \N^N$ is a partition of $n$.
Define $k(\la)=\sum_{i=1}^N \la_i(\la_i-1)/2$.

Let $v_1,\ldots,v_N$ be the standard basis in $\C^N$,
\;$e_{i,j}v_k=\delta_{\jk}v_i$. For a multi-index $L=(l_1,\ldots,l_n)$ define
$v_L=v_{l_1}\otimes \ldots \otimes v_{l_n}$. A $V[\la]$-valued function $I$ can be expressed as
\be
I=\sum_L f_L(z_1,\ldots,z_n,h)\,v_L
\ee
for multi-indices $L=(l_1,\ldots,l_n)$ with $|\{j:l_j=i\}|=\la_i$.
Denote the multi-index
\beq
\label{eq:standardL}
{L_0}=(\underbrace{1,\ldots,1}_{\la_1},\underbrace{2,\ldots,2}_{\la_2},
\ldots,\underbrace{N,\ldots,N}_{\la_N}).
\eeq

In what follows we will be concerned with the degree of polynomials and
rational functions in $z_i$, $h$. Our convention is that $\deg z_i=\deg h=1$.
With this convention, the deformed $S_n$-action of Section \ref{sec:ops} is
homogeneous. Hence, if $I$ is a skew-symmetric $V[\la]$-valued polynomial, then
its homogeneous parts are also such. Now we study what the homogeneous degrees
of skew-symmetric polynomials can be.

\smallskip
Define another $S_n$-action on functions of $z_1\lc z_n$, where the $i$-th
elementary transposition $s_i\in S_n$ is acting by the formula
\begin{gather}
\label{Snh}
s_i:f\,\mapsto\,\hat s_i f\,,
\\
\hat s_if(z_1\lc z_n)\,=\,\frac{z_i-z_{i+1}+h}{z_i-z_{i+1}}\;
f(\ldots,z_{i+1},z_i,\ldots)-\frac{h}{z_i-z_{i+1}}\;
f(\ldots,z_i,z_{i+1},\ldots)\,.
\notag
\end{gather}
For a permutation $\si\in S_n$ and a multi-index $L=(l_1\lc l_n)$ set
$\si(L)=(l_{\si^{-1}(1)}\lc l_{\si^{-1}(n)})$. The following lemma is obvious.

\begin{lem}
\label{fL}
A $V[\la]$-valued function $I$ is skew-symmetric with respect to
action~\eqref{Sn} if and only if \;$f_{s_i(L)}=-\>\hat s_if_L$ \;for every
multi-index $L$ and every \,$i=1\lc n-1$.
\end{lem}

\begin{lem}\label{lem:qdet}
If a polynomial $f(z_1,\ldots,z_n)$ is skew-symmetric with respect to
the $S_n$-action~\eqref{Snh}, then it is divisible by
$$
\prod_{1\leq i<j\leq n} (z_i-z_j+h).
$$
\end{lem}

\begin{proof} Skew-symmetry with respect to $\hat s_{i}$ implies
\begin{equation}
\label{eq:div}
(z_i-z_{i+1}+h)\,f(\ldots,z_{i+1},z_i,\ldots)\,=\,
(z_{i+1}-z_i+h)\,f(\ldots,z_i,z_{i+1},\ldots)\,.
\end{equation}
Therefore $z_i-z_{i+1}+h$ divides $f$.
This further implies that $z_{i-1}-z_{i+1}+h$ divides
$f(\ldots,\alb z_i,\alb z_{i-1},\alb \ldots)$, which using \eqref{eq:div}
again yields that $z_{i-1}-z_{i+1}+h$ divides $f$. Iterating this idea we
obtain the statement of the lemma.
\end{proof}

\begin{lem} \label{lem:deg}
\strut
\begin{enumerate}
\item[(i)] If $I\ne 0$ is a $V[\la]$-valued skew-symmetric polynomial,
then its degree is at least $k(\la)$.
\item[(ii)] A $V[\la]$-valued skew-symmetric polynomial of homogeneous degree
$k(\la)$ is unique up to multiplication by a number.
\item[(iii)] There exists a nonzero $V[\la]$-valued skew-symmetric polynomial
of homogeneous degree $k(\la)$.
\end{enumerate}
\end{lem}

\begin{proof}
(i)\enspace
By Lemma~\ref{fL}, a $V[\la]$-valued skew-symmetric polynomial $I$ is uniquely
determined by the coefficient $f_{L_0}\,$, and \,$\deg\,I=\deg f_{L_0}\,$.

Denote by $S_{\la_1}\times\ldots\times S_{\la_N}\subset S_n$
the isotropy subgroup of $L_0\,$. By Lemma~\ref{fL}, for
$s_i\in S_{\la_1}\times\ldots\times S_{\la_N}$ we have
\;$\hat s_if_{L_0}=-f_{L_0}$.
Using Lemma \ref{lem:qdet} for each $S_{\la_i}$ we obtain that $f_{L_0}$ is divisible by
\beq
\label{eq:standardD}
D_0\,=\,\prod_{1\le a<b\le \la_1} (z_a-z_b+h)
\prod_{\la_1< a<b\le \la_1+\la_2} (z_a-z_b+h) \cdots
\prod_{n-\la_N< a<b\le n} (z_a-z_b+h)
\eeq
and has degree at least $k(\la)$.

\smallskip
(ii)\enspace
If $f_{L_0}$ has degree $k(\la)$, then it is proportional to $D_0$.

\smallskip
(iii)\enspace
Define
\beq
\label{eq:Ilambda}
I_\la\,=\sum_{\si\in S_n/S_{\la_1}\times\ldots\times S_{\la_N}}\!\!
\sgn(\si)\,\hat\si(D_0)\,v_{\si(L_0)}\,,
\eeq
where \,$\sgn(\si)$ is the sign of the shortest permutation in the coset,
and \,$\hat\si$ \,denotes action~\eqref{Snh}. Then $I_\la$ is a nonzero
$V[\la]$-valued skew-symmetric polynomial of homogeneous degree $k(\la)$.
\end{proof}

The $V[\la]$-valued polynomial $I_\la$, defined by \eqref{eq:Ilambda},
is the main object of this paper. Now we reformulate its definition.

\begin{defn}
Let $\la\in\N^N$ be a partition of $n$. Let $I_\la$ be
the $V[\la]$-valued skew-symmetric polynomial of degree $k(\la)$
normalized in such a way that the coefficient of $v_{L_0}$ is~$D_0$,
see \eqref{eq:standardL}, \eqref{eq:standardD}.
\end{defn}

\begin{example} \rm We have
\begin{align*}
I_{(1,1)}={}& \>v_{12}-v_{21}\,,
\\[4pt]
I_{(2,1)}={}& \>(z_1-z_2+h)v_{112}+(z_3-z_1-2h)v_{121}+(z_2-z_3+h)v_{211}\,,
\\[4pt]
I_{(2,2)}={}& \>(z_1-z_2+h)(z_3-z_4+h)(v_{1122}+v_{2211}) + (z_1-z_4+2h)(z_2-z_3+h)(v_{1221}+v_{2112})
\\
&{}+ \bigl(-(z_1-z_2+h)(z_3-z_4+h)-(z_1-z_4+2h)(z_2-z_3+h)\bigr)
(v_{1212}+v_{2121}).
\end{align*}
Note that the last coefficient function in the third formula does not factor.
\end{example}

\begin{rem}
\label{rem:quasiclass}
In the quasiclassical limit $h=0$, the vector $I_\la$ is
the minimal degree skew-symmetric polynomial under the $S_n$-action
\be
s_i^{h=0}: I\,\mapsto\,P^{(\ii+1)} I(\ldots,z_{i+1},z_i,\ldots).
\ee
Its explicit form is
\beq\label{Ihzero}
I_\la^{h=0}=\,
\sum_L\,\Bigl(\sgn(L)\!\!\!\prod_{a<b, l_a=l_b} \!\!\! (z_a-z_b)\Bigr)\,v_L\,,
\eeq
with an appropriately defined \,$\sgn(L)$. Rescaled by the the discriminant
\,$\prod_{a<b}(z_a-z_b)$\,, this function was studied in \cite{RV}, see also
\cite{RSV}. It is shown there that this function satisfies (non-quantum)
conformal block properties and (non-quantum) \KZ/ differential equations.
\end{rem}

%
%
%
\newdimen\linkpatternunit%
\newcount\linkpatternsize%
\newcount\lpsize
%
\newif\iflinkpatterninverted%
\newif\iflinkpatternhalfnumbered
\newif\iflinkpatternnumbered%
\newif\iflinkpatterntikzstarted%
\pgfkeys{/linkpattern/.cd,inverted/.is if=linkpatterninverted,numbered/.is if=linkpatternnumbered,tikzstarted/.is if=linkpatterntikzstarted,halfnumbered/.is if=linkpatternhalfnumbered,vertexcolor/.store in=\linkpatternvertexcolor,edgecolor/.store in=\linkpatternedgecolor,size/.code={\linkpatternsize=#1},unit/.code={\linkpatternunit=#1},unpaired length/.store in=\linkpatternunpairedlength}
\linkpatterninvertedfalse%
\linkpatternnumberedfalse%
\linkpatternhalfnumberedfalse%
\linkpatterntikzstartedfalse%
\linkpatternunit=0.65cm%
\linkpatternsize=0%
\def\linkpatternvertexcolor{red}%
\def\linkpatternedgecolor{blue}%
\def\linkpatternunpairedlength{0.5}%
\newcommand\linkpattern[2][]{%
{
\pgfkeys{/linkpattern/.cd,#1}
\iflinkpatternnumbered%
\tikzset{vertex/.style={circle,draw=black,fill=\linkpatternvertexcolor,inner sep=1.5pt,label=\iflinkpatterninverted above\else below\fi:{$\scriptstyle ##1$}}}%
\else%
\tikzset{vertex/.style={circle,draw=black,fill=\linkpatternvertexcolor,inner sep=1.5pt}}%
\fi%
\tikzset{edge/.style={bend \iflinkpatterninverted right\else left\fi=60,draw,very thick,\linkpatternedgecolor}}%
\iflinkpatterntikzstarted\else%
\iflinkpatterninverted%
\begin{tikzpicture}[x=\linkpatternunit,y=-\linkpatternunit,baseline=0]%
\else%
\begin{tikzpicture}[x=\linkpatternunit,y=\linkpatternunit,baseline=0]%
\fi%
\fi%
\global\lpsize=\linkpatternsize
\foreach \x/\y in {#2}
{
\ifnum\x>\lpsize\global\lpsize=\x\fi
\ifnum\y>\lpsize\global\lpsize=\y\fi
}
%
\draw (1,0) -- (\lpsize,0);
\foreach\x in {1,...,\lpsize}
\coordinate[vertex={\x}] (v\x) at (\x,0);
\foreach \x/\y in {#2}
{
\ifnum\x<\y
\draw[edge] (v\x) to (v\y);
\else
\ifnum\x>\y
\draw[edge] (v\y) to (v\x);
\else
\draw[edge] (v\x) -- ++(0,\linkpatternunpairedlength);
\fi
\fi
}
\foreach \x/\y in {#2}
{
\ifnum\x=\y
\draw[edge] (v\x) -- (v\x |- current bounding box.north west);
\fi
}
\iflinkpatterntikzstarted\else%
\end{tikzpicture}%
\fi%
}}%
%
\newcommand\tangle[2][]{%
{
\pgfkeys{/linkpattern/.cd,#1}
\def\internalwithoutprime##1'{##1}%
\def\withoutprime##1{\expandafter\internalwithoutprime##1}%
\def\primetest##1'{}%
\def\hasaprime##1{\expandafter\primetest##1''}
\iflinkpatternnumbered%
\iflinkpatternhalfnumbered%
\tikzset{vertexa/.style={circle,draw=black,fill=\linkpatternvertexcolor,inner sep=1.5pt}}%
\else%
\tikzset{vertexa/.style={circle,draw=black,fill=\linkpatternvertexcolor,inner sep=1.5pt,label=\iflinkpatterninverted below\else above\fi:{$\scriptstyle ##1$}}}%
\fi
\tikzset{vertexb/.style={circle,draw=black,fill=\linkpatternvertexcolor,inner sep=1.5pt,label=\iflinkpatterninverted above\else below\fi:{$\scriptstyle ##1$}}}%
\else%
\tikzset{vertexa/.style={circle,draw=black,fill=\linkpatternvertexcolor,inner sep=1.5pt}}%
\tikzset{vertexb/.style={circle,draw=black,fill=\linkpatternvertexcolor,inner sep=1.5pt}}%
\fi%
\tikzset{edge/.style={draw,very thick,\linkpatternedgecolor,bend left=45}}%
\iflinkpatterntikzstarted\else%
\iflinkpatterninverted%
\begin{tikzpicture}[x=\linkpatternunit,y=-\linkpatternunit,baseline=0.5\linkpatternunit]%
\else%
\begin{tikzpicture}[x=\linkpatternunit,y=\linkpatternunit,baseline=-0.5\linkpatternunit]%
\fi%
\fi%
\global\lpsize=\linkpatternsize
\foreach \x/\y in {#2}%
{%
\if\hasaprime\x %
\ifnum\lpsize<\x\global\lpsize=\expandafter\withoutprime\x\fi%
\else%
\ifnum\x>\lpsize\global\lpsize=\x\fi%
\fi%
\if\hasaprime\y %
\ifnum\lpsize<\y\global\lpsize=\expandafter\withoutprime\y\fi%
\else%
\ifnum\y>\lpsize\global\lpsize=\y\fi%
\fi%
}%
\draw (1,-1) -- (\lpsize,-1);
\draw (1,0) -- (\lpsize,0);
\foreach\x in {1,...,\lpsize}
{
\node[vertexa={\bar\x}] (\x') at (\x,0) {};
\node[vertexb={\x}] (\x) at (\x,-1) {};
}
\foreach \x/\y in {#2}
\draw[edge] (\x) .. controls ($(0,-0.5)!(\x)!(1,-0.5)$) and ($(0,-0.5)!(\y)!(1,-0.5)$) .. (\y);
\iflinkpatterntikzstarted\else%
\end{tikzpicture}%
\fi%
}}%
%
%
\newcommand\circlelinkpattern[2][]{%
{
\pgfkeys{/linkpattern/.cd,#1}
\iflinkpatternnumbered%
\tikzset{vertex/.style 2 args={circle,draw=black,fill=\linkpatternvertexcolor,inner sep=1.5pt,label={[shift={($(360/##2*##1:0.3)$)}]center:$\scriptstyle ##1$}}}
\else%
\tikzset{vertex/.style={circle,draw=black,fill=\linkpatternvertexcolor,inner sep=1.5pt}}%
\fi%
\tikzset{edge/.style={draw,very thick,\linkpatternedgecolor}}%
\iflinkpatterntikzstarted\else%
\iflinkpatterninverted%
\begin{tikzpicture}[x=\linkpatternunit,y=-\linkpatternunit,baseline=0]%
\else%
\begin{tikzpicture}[x=\linkpatternunit,y=\linkpatternunit,baseline=0]%
\fi%
\fi%
\global\lpsize=\linkpatternsize
\edef\mylist{#2}
\foreach \x/\y in \mylist
{
\ifnum\x>\lpsize\global\lpsize=\x\fi
\ifnum\y>\lpsize\global\lpsize=\y\fi
}
%
%
\draw (0,0) circle (1);
\foreach\x in {1,...,\lpsize}
{
\coordinate[vertex={\x}{\the\lpsize}] (v\x) at ($(360/\the\lpsize*\x:1)$);
}
\foreach \x/\y/\z in \mylist
{
\ifx\y\z%
\draw[edge] (v\x) .. controls ($0.5*(v\x)$) and  ($0.5*(v\y)$) .. (v\y);
\else
\draw[edge,decoration={markings,mark = at position 0.5 with { \arrow[semithick]{\z} }},postaction={decorate}] (v\x) .. controls ($0.5*(v\x)$) and  ($0.5*(v\y)$) .. (v\y);
\fi
}
\iflinkpatterntikzstarted\else%
\end{tikzpicture}%
\fi%
}}%
%
\newcommand\tanglelinkpattern[3][]{%
{
\pgfkeys{/linkpattern/.cd,#1}
\iflinkpatterninverted%
\begin{tikzpicture}[x=\linkpatternunit,y=-\linkpatternunit,baseline=0]%
\else%
\begin{tikzpicture}[x=\linkpatternunit,y=\linkpatternunit,baseline=0]%
\fi%
\tangle[#1,tikzstarted,halfnumbered]{#2}
\linkpattern[#1,tikzstarted,numbered=false]{#3}
\end{tikzpicture}%
}}

%
\newdimen{\cellsize}
\newcommand\bigboxes{\setlength{\cellsize}{18pt}\def\boxformat{}}
\newcommand\medboxes{\setlength{\cellsize}{14pt}\def\boxformat{}}
\newcommand\smallboxes{\setlength{\cellsize}{8pt}\def\boxformat{\scriptstyle}}
\medboxes
\newsavebox{\cellcontent}
\def\hidehrule#1#2{\kern-#1
  \hrule height#1 depth#2 \kern-#2 }%
\def\hidevrule#1#2{\kern-#1{\dimen\cellcontent=#1%
    \advance\dimen\cellcontent by#2\vrule width\dimen\cellcontent}\kern-#2 }%
\def\makeblankbox#1#2{\hbox{\lower\dp\cellcontent\vbox{\hidehrule{#1}{#2}%
    \kern-#1 
    \hbox to \wd\cellcontent{\hidevrule{#1}{#2}%
      \raise\ht\cellcontent\vbox to #1{}
      \lower\dp\cellcontent\vtop to #1{}
      \hfil\hidevrule{#2}{#1}}%
    \kern-#1\hidehrule{#2}{#1}}}}
\newcommand\cellify[1]{\defaultcell%
\sbox{\cellcontent}{\vbox to \cellsize{%
\vfill%
\hbox to \cellsize{\hfill$\boxformat #1$\hfill}
\vfill}}%
\rlap{\drawnbox}
\usebox{\cellcontent}}
\newcommand\tableau[1]{\vtop{\let\\\cr
\baselineskip -16000pt \lineskiplimit 16000pt \lineskip 0pt
\ialign{&\cellify{##}\cr#1\crcr}}}
\newcommand\defaultcell{\gdef\drawnbox{
\makeblankbox{0.2pt}{0.2pt}
}}
\newcommand\graycell{\gdef\drawnbox{%
\rlap{\color{Gray}\vrule width \cellsize height \cellsize}%
\makeblankbox{0.2pt}{0.2pt}
}}
\newcommand\thickcell{\gdef\drawnbox{
\makeblankbox{0.2pt}{0.1\cellsize}%
}}
\newcommand\missingcell{\gdef\drawnbox{}}
\newcommand\vdotscell{\gdef\drawnbox{\kern-1.6pt\vbox{\baselineskip=4pt\lineskiplimit=0pt\hbox{}\hbox{.}\hbox{.}\hbox{.}\hbox{}}}}
\newcommand\hdotscell{\gdef\drawnbox{\vbox to \cellsize{\hbox{\kern1pt$\ldotp\ldotp\ldotp$}}}}
\newcommand\vhdotscell{\gdef\drawnbox{\rlap{\kern-1.6pt\vbox{\baselineskip=4pt\lineskiplimit=0pt\hbox{}\hbox{.}\hbox{.}\hbox{.}\hbox{}}}\vbox to \cellsize{\hbox{\kern1pt$\ldotp\ldotp\ldotp$}}}}
%
%
\smallboxes 
\newcommand\rk[1]{{\smaller\bf [#1]}}
\section{Geometric description of $I_\la$}\label{sec:geom}
In this section we provide the connection of the $I_\la$
to geometry which was advertised in the introduction.

\newcommand\codim{\mathop{\mathrm{codim}}\nolimits}
\newcommand\n{\frak{n}}
\renewcommand\O{\mathcal{O}}
\newcommand\g{\frak{g}}
\newcommand\SYT{\mathrm{SYT}(\lambda)}
\newcommand\tl[1][{i,i+1}]{E^{(#1)}}
\subsection{Orbital varieties and Joseph representation}
\subsubsection{Orbital varieties}
Consider some conjugacy class of nilpotent elements inside 
$\g=\frak{gl}_n$.
Such a conjugacy class is characterized by the
unordered set of sizes of the Jordan blocks,
which form a {\em partition}\/ $\lambda'=(\lambda'_1\ge\lambda'_2\ge\cdots\ge\lambda'_K)$ of $n$. 
It is more convenient to use instead of $\lambda'$ its {\em conjugate partition}
$\lambda=(\lambda_1\ge\lambda_2\ge\cdots\ge \lambda_N)$.
If we depict partitions as Young diagrams, then the diagram of $\lambda$
is the transpose of that of $\lambda'$: the lengths of its {\em columns}\/
are the sizes of Jordan blocks. For example, 
for one block of size $3$ and one block of size $1$,
we use $\lambda=(2,1,1)$, 
that is $\tableau{&\\ \\ \\}$.

Let $\bar D_\la\subset \frak g$ be the closure of the conjugacy class
$D_\la$ associated to the partition $\lambda$.
$\bar D_\lambda$ is known to be an irreducible algebraic variety, but 
if we denote by $\n$ the space of strict upper triangular matrices, then
the intersection $\O_\lambda:=\bar D_\lambda\cap \n$ is in general reducible:
its geometric components (i.e., reduced irreducible components)
are called {\em orbital varieties}.

Given an element of
$x\in \O_\lambda$, note that $x$ leaves stable the natural
flag $0\subset \mathbb{C}\subset\mathbb{C}^2\subset\cdots\subset\mathbb{C}^n$
associated to the standard basis. So the restriction of $x$ to $\mathbb{C}^i$,
$i=0,\ldots,n$, is a nilpotent element to which can be attached
a partition of $i$ as described above, say $\varphi_i(x)$. Note that generically,
$\varphi_n(x)=\lambda$. The following results were found:
\begin{thm*}[Spaltenstein \cite{Spal}]\label{thm:spal}
Let $\lambda$ be a partition of $n$, and $x\in\O_\lambda$.
\begin{itemize}
\item
The sequence $\varphi_i(x)$ forms an increasing chain of Young diagrams, 
so there is a map $\varphi$
from $\O_\lambda$ to
the set of {\em standard Young tableaux}\/ with $n$ boxes 
(i.e., fillings of Young diagrams
with numbers $\{1,\ldots,n\}$ which are increasing along rows and columns)
such that the subdiagram of $\varphi(x)$ made of the boxes labelled from $1$ to $i$
is $\varphi_i(x)$.
\item
The irreducible components $\O_{\lambda;\alpha}$ 
of $\O_\lambda$ are the closures
of $\varphi^{-1}(\alpha)$, where $\alpha$ runs over $\SYT$,
the set of standard Young tableaux
of shape $\lambda$.
\item
The $\O_{\lambda;\alpha}$ all have the same dimension which is one half
of that of $D_\lambda$.
\end{itemize}
\end{thm*}
The dimension of $D_\lambda$ is easily calculated by computing
the stabilizer of any element of the orbit and
its dimension $\sum_{i,j}\min(\lambda'_i,\lambda'_j)=\sum_i \lambda_i^2$
(cf \cite[p.~11]{Humph-conjclass}), 
so that we find:
\begin{equation}\label{dimO}
\dim\O_{\lambda;\alpha}=\frac{n(n-1)}{2}-\frac{1}{2}\sum_{i=1}^n \lambda_i(\lambda_i-1)
\,.
\end{equation}

When there is no risk of confusion, we shall drop the index
$\lambda$: $\O_{\lambda,\alpha}=\O_\alpha$.

\subsubsection{The case $\lambda_1-\lambda_N\le 1$}
Define the {\em dominance order}\/ on partitions by
$\lambda\prec\mu$ iff $\sum_{i\le k}\lambda_i\le \sum_{i\le k}\mu_i$ for all
$k$. Then one has \cite[p.~139]{Humph-conjclass} 
$D_\mu\subset \bar D_\la$ iff $\lambda\prec\mu$.
The next proposition gives
a more explicit description of $\bar D_\la$ and $\O_\la$
in a special case which is important for our purposes:
\begin{prop}\label{prop:orbrect}
Let $\lambda=(\lambda_1\ge\cdots\ge\lambda_N)$ 
be a partition such that $\lambda_1-\lambda_N\le 1$.
Then
\begin{align*}
\bar D_\la&=\{x\in\g: x^N=0\}
\,,
\\
\O_\la&=\{x\in\n: x^N=0\}
\,.
\end{align*}
\end{prop}
\begin{proof}
According to the discussion above,
the first equality amounts to saying that among all the partitions
$\mu$ of $n$ with at most $N$ parts, $\lambda$ is the {\em smallest}\/ for
the dominance order. By direct computation, if $n=Nq+r$,
$\lambda=(\underbrace{q+1,\ldots,q+1}_r,\underbrace{q,\ldots,q}_{N-r})$;
and assuming a $\mu=(\mu_1,\ldots,\mu_N)$, $\sum_i \mu_i=n$,
breaks one of the inequalities
$\sum_{i\le k} \lambda_i\le\sum_{i\le k} \mu_i$ leads to a contradiction
with $\mu$ being decreasing.

The second equality follows immediately from the first.
\end{proof}

\newcommand\zero{\cdot}
\newsavebox{\tababox}
\savebox{\tababox}{\tableau{1\\2}}
\newcommand\taba{{\usebox{\tababox}}}
\newsavebox{\tabbbox}
\savebox{\tabbbox}{\tableau{1&2\\3}}
\newcommand\tabb{{\usebox{\tabbbox}}}
\newsavebox{\tabcbox}
\savebox{\tabcbox}{\tableau{1&3\\2}}
\newcommand\tabc{{\usebox{\tabcbox}}}
\newsavebox{\tabdbox}
\savebox{\tabdbox}{\tableau{1&2\\3&4}}
\newcommand\tabd{{\usebox{\tabdbox}}}
\newsavebox{\tabebox}
\savebox{\tabebox}{\tableau{1&3\\2&4}}
\newcommand\tabe{{\usebox{\tabebox}}}
\begin{example}
If $\lambda=(1,1)$ there is only one tableau $\taba$,
and $\O_\taba=\n=\left\{\begin{pmatrix}\zero&\star\\ \zero&\zero\end{pmatrix}\right\}$, where $\star$ denotes a free entry and $\zero$ a zero in the lower triangle.

Next,
\begin{multline*}
\O_{(2,1)}=\left\{\begin{pmatrix}\zero&x_{12}&x_{13}\\\zero&\zero&x_{23}\\\zero&\zero&\zero
\end{pmatrix}: x_{12}x_{23}=0\right\}
=
\O_\tabb\cup \O_\tabc\\
\O_\tabb=
\left\{\begin{pmatrix}\zero&0&\star\\\zero&\zero&\star\\\zero&\zero&\zero
  \end{pmatrix}\right\}
\qquad
\O_\tabc
=
\left\{\begin{pmatrix}\zero&\star&\star\\\zero&\zero&0\\\zero&\zero&\zero
  \end{pmatrix}\right\}
\,.
\end{multline*}
Similarly, one computes 
\begin{multline*}
\O_{(2,2)}=\{x\enskip 4\times 4: x^2=0\}
=\O_{\tabd}\cup\O_{\tabe}
\\
\O_{\tabd}=
\left\{\begin{pmatrix}\zero&0&\star&\star\\\zero&\zero&\star&\star\\\zero&\zero&\zero&0\\\zero&\zero&\zero&\zero
  \end{pmatrix}\right\}
\qquad
\O_{\tabe}
=
\left\{\begin{pmatrix}\zero&x_{12}&x_{13}&\star\\\zero&\zero&0&x_{24}\\\zero&\zero&\zero&x_{34}\\\zero&\zero&\zero&\zero
  \end{pmatrix}: x_{12}x_{24}+x_{13}x_{34}=0\right\}
\,.
\end{multline*}
\end{example}

\subsubsection{Hotta's construction of the Joseph representation}
\label{sec:hotta}
In the rest of Section~\ref{sec:geom} we fix a partition $\la$.
Define $W_\lambda$ to be the finite-dimensional space of maps
from $\SYT$ to $\mathbb{C}$.
Its dimension is that of the irreducible representation of $S_n$ associated
to $\la$. Orbital varieties provide us with 
a natural action of $S_n$ on $W_\la$,
which we describe now following \cite{Hotta}.

Given $i=1,\ldots,n-1$, define $\n_i$ to be the subspace
of $x\in\n$ whose entry $x_{i,i+1}$ vanishes; and $P_i$ to be
the parabolic subgroup of $GL_n$ made of invertible matrices $x$
which are upper triangular {\em except}\/ possibly at $x_{i+1,i}$.
Note that the map $f: P_i \times \g \to \g$, $f(p,x)=pxp^{-1}$
sends $P_i\times\n_i$ to $\n_i$.
We shall now describe the action of the elementary
transposition $(i,i+1)\in S_n$ by giving its matrix elements.

Given a $\alpha\in \SYT$, two situations can occur:
\begin{enumerate}
\item Either $\O_\alpha \subset \n_i$, in which case
set $m_{i;\alpha,\beta}=-\delta_{\alpha,\beta}$ for all $\beta$.
\item Or $\O_\alpha \not\subset \n_i$, in which case consider
the scheme-theoretic intersection (i.e., with multiplicities)
$\O_\alpha \cap \n_i$, and then its image
by $f$, i.e.,
$f(P_i\times(\O_\alpha \cap \n_i))$
(again keeping track of the degree of the map on each
irreducible component of $\O_\alpha\cap\n_i$);
Clearly $f(P_i\times(\O_\alpha \cap \n_i))\subset \O_\la \cap \n_i$, 
so its top-dimensional components
are again orbital varieties (necessarily distinct from $\alpha$).
Then set 
\[
m_{i;\alpha,\beta}=\begin{cases}1&\beta=\alpha\\
\text{multiplicity of }\O_{\beta}\text{ in }
f(P_i\times(\O_\alpha \cap \n_i))&\beta\ne\alpha
\end{cases}
\]
\end{enumerate}
Finally, if $(e_\alpha)_{\alpha\in\SYT}$ is the standard basis of $W_\lambda$:
$e_{\beta}(\alpha)=\delta_{\alpha,\beta}$,
then define
\begin{equation}\label{exchdual}
\rho^{(i,i+1)} e_{\beta} = -\sum_{\alpha\in\SYT} m_{i;\alpha,\beta}e_{\alpha}
\,.
\end{equation}

\begin{thm}[Hotta]\label{thm:hotta}
The $(m_{i;\alpha,\beta})_{\alpha,\beta\in\SYT}$, $i=1,\ldots,n-1$,
satisfy the symmetric group relations;
and equipped with the action $\rho^{(i,i+1)}$ above, $W_\lambda$ is the standard
$S_n$-module associated to the partition $\la$.
\end{thm}
\let\oldbordermatrix\bordermatrix
\def\bordermatrixint#1\end{\oldbordermatrix{#1}\end}
\renewenvironment{bordermatrix}{\let\\\cr
\bordermatrixint}{}
\begin{example}
For the three cases $(1,1)$, $(2,1)$, $(2,2)$, we find:
\begin{alignat*}{3}
\lambda&=(1,1):&\qquad\rho^{(1,2)}&=
{\begin{bordermatrix}&\taba\cr\taba&-1\end{bordermatrix}}
\,,
\\
\lambda&=(2,1):&\qquad\rho^{(1,2)}&=
{\begin{bordermatrix}&\tabb&\tabc\cr
\tabb&1&0\\ \tabc&-1&-1
\end{bordermatrix}}
&\quad\rho^{(2,3)}&=
{\begin{bordermatrix}&\tabb&\tabc\cr
\tabb&-1&-1\\\tabc&0&1\end{bordermatrix}}
\,,
\\
\lambda&=(2,2):&\qquad\rho^{(1,2)}&=\rho^{(3,4)}=
{\begin{bordermatrix}&\tabd&\tabe\cr
\tabd&1&0\\\tabe&-1&-1\end{bordermatrix}}
&\quad\rho^{(2,3)}&=
{\begin{bordermatrix}&\tabb&\tabc\cr
\tabd&-1&-1\\\tabe&0&1\end{bordermatrix}}
\,.
\end{alignat*}
\end{example}

\subsection{Extended Joseph polynomials}
Now consider the (complex) torus $T=(\mathbb{C}^\times)^{n+1}$
acting on $\n$ as follows: the first $n$ variables correspond to
conjugation by diagonal matrices, whereas the last variable corresponds to
scaling. Explicitly, if $x\in \n$ has entries $x_{ij}$ and $t=(t_1,\ldots,t_n,q)\in T$, then $(t\cdot x)_{ij}=q\,t_i t_j^{-1} x_{ij}$.

Observe that $\O_\lambda$, and therefore its irreducible components
$\O_{\lambda,\alpha}$, are invariant by $T$-action. Thus, they have
natural Poincar\'e-dual classes in equivariant cohomology. It is convenient
to describe them in the language of multidegrees (see \cite{MS-book}).

\subsubsection{Multidegrees}
\newcommand\mdeg{\mathop{\rm mdeg}\nolimits}

Given a  torus $T$ acting 
linearly on a complex vector space $W$, we assign
to a closed $T$-invariant sub-scheme $X\subseteq W$ 
its {\em multidegree}\/ $\mdeg_W X \in \mathrm{Sym}(T^*)$ 
(here $T^*$ is viewed as a lattice
inside the dual of the Lie algebra of $T$),
which can be computed inductively using the following properties 
(as in \cite{Jos-mdeg}):
\begin{enumerate}
\item If $X=W=\{0\}$, then $\mdeg_W X = 1$.
\item If the scheme $X$ has top-dimensional components $X_i$, 
  where $m_i>0$ denotes the multiplicity of $X_i$ in $X$, 
  then $\mdeg_W X = \sum_i m_i\ \mdeg_W X_i$. 
\item Assume $X$ is a variety, 
  and $H$ is a $T$-invariant hyperplane in $W$.
  \begin{enumerate}
  \item If $X\not\subset H$, then $\mdeg_W X = \mdeg_H (X\cap H)$.
  \item If $X\subset H$, then 
    $ \mdeg_W X = [W/H]_T \mdeg_H X$,
where $[\cdot]_T\in T^\star$ denotes the weight of the $T$-action.
  \end{enumerate}
\end{enumerate}
One can readily see from these properties that $\mdeg_W X$ is 
homogeneous of degree $\codim_W X$, and is a positive sum of products
of the weights of $T$ on $W$. 

In our case, $\mathrm{Sym}(T^\star)\cong \mathbb{Z}[z_1,\ldots,z_n,h]$,
and the weights on $\mathbb{C}[x_{ij}]_{1\le i<j\le n}$ of the $T$-action
are defined by
\[
[x_{ij}]_T=h+z_i-z_j
\]

The multidegree of $\O_{\alpha}$ with respect to this $T$-action,
$J_\alpha:=\mdeg_\n \O_{\alpha}$, is called the {\em extended Joseph
polynomial}\/ of $\O_{\alpha}$. $J_\alpha$ is by definition
a homogeneous polynomial in $\mathbb{Z}[z_1,\ldots,z_n,h]$,
of degree the
codimension of $\O_{\alpha}$,
which is nothing but $k(\lambda)=\frac{1}{2}\sum_i \lambda_i(\lambda_i-1)$ 
defined in Section~\ref{sec:Ila}, according to Eq.~\eqref{dimO}.
The reason for the name, first
given in \cite{artic34}, is that
if we remove the scaling equivariance, i.e., set the variable $h=0$,
these polynomials reduce to the ones
Joseph introduced in \cite{Jos-orbvar}.

\begin{example}
All the examples of orbital varieties given above are complete intersections
(the number of equations is equal to the codimension); their multidegree
is therefore the product of weights of the equations:
\begin{align*}
J_{\taba}&=1
\,,
\\
J_{\tabb}&=h+z_1-z_2
\,,
\\
J_{\tabc}&=h+z_2-z_3
\,,
\\
J_{\tabd}&=(h+z_1-z_2)(h+z_3-z_4)
\,,
\\
J_{\tabe}&=(h+z_2-z_3)(2h+z_1-z_4)
\,.
\\
\end{align*}
\end{example}

\subsubsection{Divided Differences}\label{sec:divdiff}
The geometric construction given in Section~\ref{sec:hotta} has a direct
counterpart for multidegrees. Here our reference is \cite[Sect.~5.1.1]{artic39}.

Define the {\em divided difference operator}\/ $\der_i=\frac{1}{z_i-z_{i+1}}(\tau_i-1)$ where $\tau_i$ is permutation of variables $z_i$ and $z_{i+1}$.
Note that both $\der_i$ and $\tau_i$ are operators leaving
$\mathbb{Z}[z_1,\ldots,z_n]$ stable.

Let $B$ be the group of invertible upper triangular matrices of size $n$.
We use the following special case of \cite{Jos-orbvar} (see also \cite[Lemma 8]{artic39}):
\begin{lem}\label{lem:bs}
Let $X\subset\n$ be a $B$-invariant variety such
that $f(P_i\times X)\subset \n$. Let $k$ be the
degree of the map $f|_X: (P_i\times X)/B \to \n$, or zero if the generic
fiber is infinite (i.e., $X$ is $P_i$-invariant). Then
\[
-\frac{1}{h+z_{i+1}-z_i}
\der_i ((h+z_{i+1}-z_i) \mdeg_{\n}X)=k \mdeg_\n f(P_i\times X)
\,.
\]
\end{lem}
The proof is a standard equivariant localization argument which we shall
not repeat here.

We now discuss separately the two cases of the construction
of Section~\ref{sec:hotta}. Given a $\alpha\in\SYT$,
\begin{enumerate}
\item If $\O_\alpha \subset \n_i$, then $f(P_i\times \O_\alpha)\subset \n_i\cap\O_\la$,
is irreducible and contains $\O_\alpha$; therefore it is equal (set-theoretically)
to $\O_\alpha$, i.e., $\O_\alpha$ is $P_i$-invariant. Lemma~\ref{lem:bs} implies
\begin{equation}\label{excha}
\der_i ((h+z_{i+1}-z_i)J_\alpha)=0
\,.
\end{equation}
\item If $\O_\alpha\not\subset \n_i$, we have $\mdeg_\n (\O_\alpha\cap\n_i)= (h+z_i-z_{i+1})J_\alpha$ by property (3b) of multidegrees, and then by applying Lemma~\ref{lem:bs} to each irreducible component of $\O_\alpha\cap\n_i$ we find:
\begin{equation}\label{exchb}
-(h+z_i-z_{i+1}) \der_i J_\alpha = \sum_{\beta\ne\alpha} m_{i;\alpha,\beta} J_{\beta}
\,.
\end{equation}
\end{enumerate}
Adding the diagonal term to the sum,
Eq.~\eqref{exchb} can be rewritten under the equivalent form
\begin{equation}\label{exch}
\hat s_i J_\alpha =\sum_{\beta} m_{i;\alpha,\beta} J_{\beta}
\end{equation}
where we used the $S_n$-action $\hat s_i=\tau_i-h\der_i$ of Eq.~\eqref{Snh}.
Note now that Eq.~\eqref{excha} is a special case of Eq.~\eqref{exch}
where $m_{i;\alpha,\beta}=-\delta_{\alpha,\beta}$. So Eq.~\eqref{exch} is valid
in all cases. At $h=0$, which is the case that Joseph considered in \cite{Jos-orbvar},
$\hat s_i$ reduces to the action of $S_n$ on $\mathbb{C}[z_1,\ldots,z_n]$
by permutation of variables.

\subsection{Identification with $I_\lambda$}
There is a natural object in $W_\la\otimes\mathbb{C}[z_1,\ldots,z_n,h]$,
namely the map $J_\lambda: \alpha\in\SYT\mapsto J_\alpha$. Combining Eqs.~(\ref{exchdual}) and (\ref{exch}),
we find:
\begin{equation}\label{exchbis}
\rho^{(i,i+1)} J_\la = -\hat s_i J_\la
\,.
\end{equation}
According to Theorem \ref{thm:hotta}, $W_\lambda$ carries the structure
of $S_n$-module which is the same, by Schur--Weyl duality,
as that of the space of singular vectors in $V[\lambda]$
(where $S_n$ acts by permutation $P^{(i,i+1)}$ of tensors).
Tensoring with $\mathbb{C}[z_1,\ldots,z_n,h]$ (on which we do not make
$S_n$ act), we obtain an $S_n$-intertwiner $\phi:
W_\la\otimes\mathbb{C}[z_1,\ldots,z_n,h]
\to V[\la]_{sing}\otimes\mathbb{C}[z_1,\ldots,z_n,h]$.
The equation above becomes
\[
P^{(i,i+1)}\phi(J_\la)=-\hat s_i \phi(J_\la)
\]
which means $\phi(J_\la)$ satisfies the hypothesis of Lemma~\ref{fL}.
Note that $\phi$ is only defined up to a non-zero multiplicative constant.

We now want to identify $\phi(J_\la)$ with $I_\la$ by using Lemma~\ref{lem:deg}.
By definition the entries of $\phi(J_\la)$ are linear combinations
of those of $J_\la$ and therefore are homogeneous polynomials of
degree $k(\lambda)$ in the variables $z_1,\ldots,z_n,h$.
We have just derived the skew-symmetry of $\phi(J_\la)$ from Lemma~\ref{fL}.
Therefore, we have proved:
\begin{thm}\label{thm:ident}
Let $\lambda$ be a partition of $n$.
Then there exists a unique intertwiner $\phi$ such that
\[
\phi(J_\la)=I_\la
\,.
\]
\end{thm}
In particular, all properties that we shall prove for $I_\la$
are true for $J_\la$ as well.

\begin{example}
By comparing the formulae for $I_\la$ and $J_\la$, we find
\begin{alignat*}{3}
\lambda&=(1,1):&\qquad\phi(e_{\taba})&=v_{12}-v_{21}
\,.
\\
\lambda&=(2,1):&\qquad\phi(e_{\tabb})&=v_{112}-v_{121}
\,,
&\quad
\phi(e_{\tabc})&=v_{211}-v_{121}
\,.
\\
\lambda&=(2,2):&\qquad\phi(e_{\tabd})&=v_{1122}+v_{2211}-v_{1212}-v_{2121}
\,,
&\quad
\phi(e_{\tabe})&=v_{1221}+v_{2112}-v_{1212}-v_{2121}
\,.
\end{alignat*}
\end{example}

We next investigate in more detail two special cases for which
everything can be worked out explicitly; the reader is invited
to check all the results on our running examples, which belong to both.

\newcommand\A{h}

\subsection{Case of two rows}
In this section, we assume that the partition $\lambda$ has only two rows:
$\lambda=(n-p,p)$. This case was investigated in detail
in the paper \cite{artic39}, so that we shall omit proofs of the
results that were already contained in it.

\subsubsection{Link patterns} \label{sec:lp}
Call {\em link pattern} (or non-crossing matching) an unordered
collection of disjoint pairs of $\{1,\ldots,n\}$ such that
if $\{1,\ldots,n\}$ are represented as ordered vertices on a line,
then the two elements of each pair 
can simultaneously be connected in the upper half plane in a non-crossing
fashion (we sometimes call these connecting lines {\em arches})
and unpaired elements can be connected to upwards infinity (i.e.,
they must be outside all arches);
cf Fig.~\ref{fig:lp} left.
\begin{figure}
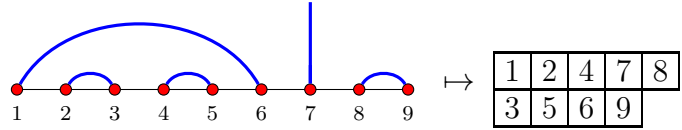

\[
\linkpattern[numbered]{1/6,2/3,4/5,7/7,8/9}
\ \mapsto\ 
\medboxes
\tableau{1&2&4&7&8\\3&5&6&9}
\]
\caption{From link patterns to standard Young tableaux.}\label{fig:lp}
\end{figure}

There is a simple bijection from link patterns with $p$ arches
to standard Young tableaux of shape $\lambda=(n-p,p)$, obtained
by recording in the first row the locations of openings of arches
and of empty spots, and on the second row the locations of
closings of arches. see Fig.~\ref{fig:lp}.

Given $\alpha\in\SYT$, we can therefore consider its associated link pattern,
and in particular, we shall use the following notation:
if $i$ and $j$ are paired in the link pattern, write $\alpha(i)=j$, $\alpha(j)=i$;
if $i$ is unpaired, write $\alpha(i)=\varnothing$.

There is an action of the {\em Temperley--Lieb algebra} on the
space of linear combinations of link patterns, which we identify
with $W_\lambda$ by identifying $e_\alpha$ with the corresponding link pattern.
It is defined graphically by
the action of the generators $\tl$, $i=1,\ldots,n-1$, 
of the Temperley--Lieb algebra which
corresponds to reconnecting vertices $i$ and $i+1$, e.g.,
\begin{align*}
\tl[1,2]\,\linkpattern[numbered]{1/1,2/2,3/4}&=
\tanglelinkpattern[numbered]{1/2,1'/2',3/3',4/4'}{1/1,2/2,3/4}=0
\\
\tl[1,2]\,\linkpattern[numbered]{1/1,2/3,4/4}&=
\tanglelinkpattern[numbered]{1/2,1'/2',3/3',4/4'}{1/1,2/3,4/4}=
\linkpattern[numbered]{1/2,3/3,4/4}
\\
\tl[1,2]\,\linkpattern[numbered]{1/4,2/3}
&=
\tanglelinkpattern[numbered]{1/2,1'/2',3/3',4/4'}{1/4,2/3}=\linkpattern[numbered]{1/2,3/4}
\\
\tl[1,2]\,\linkpattern[numbered]{1/2,3/4}
&=
\tanglelinkpattern[numbered]{1/2,1'/2',3/3',4/4'}{1/2,3/4}=2\linkpattern[numbered]{1/2,3/4}
\end{align*}
with the additional rules that reconnecting two unpaired points produces zero and that closed loops must be erased at a cost of multiplication
by $2$.

The Temperley--Lieb algebra (at loop weight $2$) is a quotient
of the symmetric group algebra, and in fact
it is shown in \cite[Sect.~5.1.2]{artic39} that if one sets
$\rho^{(i,i+1)}=1-\tl$, then the action defined above coincides
with the Joseph representation defined in Section~\ref{sec:hotta}.

\subsubsection{Description of orbital varieties}
When $\lambda$ has only two rows, $\O_\la$ is a ``spherical variety'',
i.e., the group of invertible upper triangular matrices
$B$ acts on it by conjugation
with a {\em finite}\/ number of orbits. This gives us a
first description of its irreducible components as $B$-orbit closures.
Define $\alpha_<$ to be the upper triangular matrix with values in $\{0,1\}$
such that $(\alpha_<)_{ij}=1$ iff $j=\alpha(i)>i$. Then it is easy to show
that $\O_\alpha=\overline{B\cdot \alpha_<}$, where $\cdot$ denotes conjugation
action.

An alternative description is in terms of equations:
\begin{prop}\label{prop:tworoweq}
$\O_\alpha$ is defined by the following equations:
\begin{enumerate}
\item $x^2=0$. 
\item The rank of any lower-left submatrix of $x$ is lower or equal
to the rank of the same submatrix of $\alpha_<$.
\end{enumerate}
\end{prop}
This statement can be extracted with some effort
from \cite{Meln2} or can be deduced
directly from the results of \cite{artic39}.

\subsubsection{Exchange relation}\label{sec:dicho}
We rewrite explicitly the dichotomy of the Hotta construction (Sections \ref{sec:hotta} and \ref{sec:divdiff}) since it will be needed in the next section.

Consider $\alpha\in\SYT$ and its associated link pattern.
The two cases are:
\begin{enumerate}
\item Either $\alpha(i)\ne i+1$
(there is no arch connecting $i$ and $i+1$ in the link pattern),
which means $(\alpha_<)_{i,i+1}=0$ and according to Proposition \ref{prop:tworoweq}
(2), among the equations of $\O_\alpha$ there is $x_{i,i+1}=0$, i.e.,
$\O_\alpha\subset \n_i$, and Eq.~\eqref{excha} holds,
or equivalently, $J_\alpha$ is $h+z_i-z_{i+1}$ times a symmetric polynomial
in $z_i,z_{i+1}$.
\item Or $\alpha(i)=i+1$
(there is an arch connecting $i$ and $i+1$ in the link pattern),
in which case $(\alpha_<)_{i,i+1}=1$, which implies $\O_\alpha\not\subset\n_i$.
Then we can rewrite Eq.~\eqref{exchb} as:
\[
-(h+z_i-z_{i+1}) \der_i J_\alpha = \sum_{\beta: \tl\beta=\alpha} J_{\beta}
\]
where, to keep notations simple, we have identified standard Young tableaux
and link patterns when we write ``$\tl\beta=\alpha$''.
\end{enumerate}

\subsubsection{Change of basis}
Finally, we investigate the intertwiner $\phi$.
In fact, in the case of two-row diagrams, there is a well-known
explicit formula for $\phi$, which was rediscovered
many times and dates back (at least in the special
case $p=n/2$) to \cite{RTW}
(see \cite{artic52} for more references and background). 
Roughly speaking, a pairing between $i$ and $j$ corresponds to
a $\frak{sl}(2)$ singlet $v_1^{(i)}\otimes v_2^{(j)}-v_2^{(i)}\otimes v_1^{(j)}$, 
whereas
an unpaired $i$ is a $v_1^{(i)}$ 
(where superscripts are as usual locations in the tensor
product $(\mathbb{C}^2)^{\otimes n}$).
With the present
sign conventions, the exact statement is:
\begin{prop}\label{prop:intertworow}
The intertwiner $\phi$ is given by
\[
\phi(e_\alpha)=\sum_{\substack{L=(l_1,\ldots,l_n)\in \{1,2\}^n\\l_i\ne l_j\text{ if }
j=\alpha(i)\\l_i=1\text{ if }\alpha(i)=\varnothing}} (-1)^{\lfloor \frac{n-p}{2}\rfloor+\#\{i\text{ even: }l_i=1\}}
v_L
\,.
\]
\end{prop}
\begin{proof}
Checking that the $\phi$ thus defined intertwines
the actions of the symmetric group is a routine exercise. 
Because of the conditions on the conditions on the multi-index $L$,
there are $k$ $2$'s and $n-k$ $1$'s, so $\phi(e_\alpha)\in V_\la$.
So the only
issue is normalization of $\phi$, 
which is fixed by Theorem~\ref{thm:ident}.
Consider the multi-index $L_0=(\underbrace{1,\ldots,1}_{n-p},\underbrace{2,\ldots,2}_p)$. We have
\[
I_{L_0}=D_0=\prod_{1\le i<j\le n-p}(h+z_i-z_j)\prod_{n-p+1\le i<j\le n} (h+z_i-z_j)
\,.
\]
On the other hand, by inspection the entry $\phi(e_\alpha)_{L_0}$ is zero
unless $\alpha$ is the tableau $\setlength{\cellsize}{18pt}\def\boxformat{\scriptstyle}
\alpha_0=\tableau{1&&\cdots&&n-p\\\vbox{\hbox{$\scriptstyle n-p$}\hbox{$\scriptstyle+1$}}&\cdots&n}$,
corresponding to the link pattern 
$\begin{tikzpicture}[x=\linkpatternunit,y=\linkpatternunit,baseline=0]%
\linkpattern[tikzstarted]{1/1,2/2,3/8,4/7,5/6}
\draw[decoration={brace,amplitude=7},decorate,thick]
($(v5)+(0.2,-0.1)$) -- ($(v3)+(-0.2,-0.1)$); \node[below=0.4 of v4]{$p$};
\end{tikzpicture}$. $J_{\alpha_0}$ is the multidegree of the
orbital variety indexed by $\alpha_0$, 
which according to Proposition \ref{prop:tworoweq}
is a linear subspace of the form
\newcommand{\tikzmark}[1]{\tikz[overlay,remember picture] \coordinate (#1);}
\[
\O_{\alpha_0}=\left\{\begin{pmatrix}
\cdot&0&\cdots&0&\tikzmark{nw}\star&\cdots&\cdots&\star\tikzmark{ne}\\
\cdot&\cdot&\ddots&\vdots&\vdots&&&\vdots\\
\cdot&\cdot&\cdot&0&\vdots&&&\vdots\\
\cdot&\cdot&\cdot&\cdot&\star&\cdots&\cdots&\star\tikzmark{se}\\
\cdot&\cdot&\cdot&\cdot&\cdot&0&\cdots&0\\
\cdot&\cdot&\cdot&\cdot&\cdot&\cdot&\ddots&\vdots\\
\cdot&\cdot&\cdot&\cdot&\cdot&\cdot&\cdot&0\\
\cdot&\cdot&\cdot&\cdot&\cdot&\cdot&\cdot&\cdot\\
\end{pmatrix}\qquad\quad
\begin{tikzpicture}[overlay,remember picture] 
\draw[decoration={brace,amplitude=7},decorate,thick] 
($(nw)+(0,0.2)$) -- ($(ne)+(0,0.2)$);
\node at ($(nw)!0.5!(ne)+(0,0.7)$) {$p$};
\draw[decoration={brace,amplitude=6},decorate,thick] 
($(ne)+(0.3,0.2)$) -- ($(se)+(0.3,0)$);
\node at ($(ne)!0.5!(se)+(1.2,0.1)$) {$n-p$};
\end{tikzpicture}
\right\}
\]
so that $J_{\alpha_0}=I_{L_0}$. We conclude that the normalization
of $\phi$ is fixed by $\phi(e_{\alpha_0})_{L_0}=1$.
This fits with the formula of the proposition.
\end{proof}

\subsubsection{Cyclicity}\label{sec:cycl}
Consider the special case $n=2p$,
i.e., the Young diagram is rectangular, and the corresponding
link patterns have no unpaired vertices. One can then define a rotation
of link patterns in the natural way, i.e., move the vertices
cyclically $1\to2\to\cdots\to n\to1$ keeping the pairings intact.
Via the one-to-one correspondence from link pattens to $\SYT$,
this defines a bijection $\rho$ from $\SYT$ to $\SYT$.
One observes empirically the following relation:
\begin{equation}\label{eq:tworowcyc}
J_{\rho\alpha}(z_1,\ldots,z_n)=(-1)^{p-1}J_\alpha(z_2,\ldots,z_n,z_1-3h)
\end{equation}
It is well-known that if Eq.~\eqref{exch} is satisfied, then
Eq.~\ref{eq:tworowcyc} is equivalent to the \qKZ/ equation \eqref{eq:qkz}
(in which $N=2$).
Indeed we shall prove in Section \ref{sec:proofs}
that the case $\lambda=(p,p)$ is among the
cases where $I_\lambda$ and therefore
$J_\lambda$ satisfy the \qKZ/ equation.

The case $\lambda=(p,p)$ will be considered again in Section~\ref{sec:Selberg}.

\subsection{Case of two columns}
Let $\lambda=(\underbrace{2,\ldots,2}_p\underbrace{1,\ldots,1}_{N-p})$,
$n=N+p$. In this section we show that 
all orbital varieties $\O_\al$ for such $\la$ are complete intersections and
we describe explicitly their equations as well as the extended
Joseph polynomials $J_\al$.

Note that the codimension of $\O_\la$ is simply $k(\lambda)=p$.
Furthermore, $\lambda$ satisfies the hypothesis of Proposition \ref{prop:orbrect}, so that
\[
\O_\la=\{ x\enskip n\times n: x^N=0\}\qquad n=N+p,\quad p\le N
\,.
\]

There is a general duality of orbital varieties which corresponds
to conjugation of partitions and standard Young tableaux 
(related to the duality of \cite[Chapter 3]{Spal}).
It means that the cases of two rows and two columns are dual to each other.

\newcommand\SYTp{\mathrm{SYT}(\lambda')}
\subsubsection{The dual symmetric group action}
As mentioned above, there is a bijection $'$ from $\SYT$ to
$\SYTp$ which is just conjugation (reflection through the diagonal)
of Young tableaux. Using this bijection
we shall define a new ``dual'' action on $W_\la$ starting from that
on $W_{\la'}$ as defined in Section~\ref{sec:lp}. Recall
that the action is defined by the $m_{i;\alpha,\beta}$, cf~Eq.~\eqref{exchdual}.

Given $\alpha\in\SYT$, define the sign of $\alpha$ to be
\begin{equation}\label{eq:sign}
\epsilon_\alpha=(-1)^{\textstyle \#\{i<j:\text{ $j$ strictly south-west of $i$ in $\alpha$}\}}
\,.
\end{equation}

Now given $\alpha,\beta\in\SYT$ and their conjugate $\alpha',\beta'$, define
\begin{equation}\label{eq:dualaction}
m_{i;\alpha,\beta}=-\epsilon_\alpha\epsilon_\beta
m_{i;\beta',\alpha'}
\,.
\end{equation}
Due to the easy lemma the $m_{i;\alpha,\beta}\ne0$, $\alpha\ne\beta$, implies $\epsilon_\alpha\ne\epsilon_\beta$, we can write equivalently
\[
m_{i;\alpha,\beta}=(-1)^{\delta_{\alpha,\beta}} 
m_{i;\beta',\alpha'}
\]
i.e., negate the diagonal entries and transpose.


\subsubsection{Defining equations of the orbital varieties}
There is again a bijection between $\SYT$ and the
set of link patterns in size $n$
with $p$ arches, obtained by composing the bijection of 
Section~\ref{sec:lp} with conjugation of Young tableaux. 
So we shall use the same notations $\alpha(i)=j$, $\alpha(i)=\varnothing$,
for paired $i,j$ and unpaired $i$, respectively.

For $\alpha\in\SYT$ and $1\le i<j\le 2n$, denote 
\[
p_\alpha(i,j):=j-i+1-\#\{k:\ i\le k<\alpha(k)\le j\}
\,.
\]
We have $p_\alpha(i,j)\ge\frac{j-i+1}{2}$, with equality
if and only if all elements of $[i,j]$ are paired between
themselves; in particular,
\[
p_\alpha(i,j)=\frac{j-i+1}{2}\qquad i<j=\alpha(i)
\,.
\]
An important property is that if $i\le i'\le j'\le j$, then
$p_\alpha(i',j')\le p_\alpha(i,j)$ (enlarging an interval by one can either
leave $p_\alpha$ unchanged if a new pairing has been absorbed in the interval,
or increase $p_\alpha$ by 1 otherwise).
This implies that $p_\alpha(i,j)\le p_\alpha(1,n)=n-p=N$ for all $i<j$.

Define 
\begin{equation}\label{defeq}
\hat\O_\alpha:=\left\{x\in \n:
\left(x^{p_\alpha(i,j)}\right)_{i,j}=0,\ i<j=\alpha(i)
\right\}
\,.
\end{equation}
We want to show that $\hat\O_\alpha=\O_\alpha$.
First we prove the following lemma:
\begin{lem}\label{lem:poweq}
If $x\in\hat\O_\alpha$
then
$\left(x^a\right)_{i,j}=0$
for all $i<j$ and $a\ge p_\alpha(i,j)$.
\end{lem}
In fact, these are all the
$\left(x^a\right)_{i,j}$ in the ideal of equations of $\hat\O_\alpha$.
\begin{proof}
By induction on $j-i$. 

If $j=i+1$ then either $a>1$ in which case $(x^a)_{i,i+1}=0$ because
$x$ is strict upper triangular; or $a=1=p_\alpha(i,i+1)$ which implies
$i+1=\alpha(i)$, in which case $x_{i,i+1}=0$ is part of the defining equations
of $\hat\O_\alpha$.

Next, assume $j>i+1$. We are going to divide into cases depending on the
position of $\alpha(i)$ (and similarly for $\alpha(j)$).
If $\alpha(i)=j$ and $a=p_\alpha(i,j)$, once again
$\left(x^a\right)_{i,j}=0$ is part of the defining equations
of $\hat\O_\alpha$.
If $\alpha(i)\not\in [i,j]$ or $\alpha(i)=j$ and $a>p_\alpha(i,j)$, consider
\[
\left(x^a\right)_{i,j}=
\sum_{i<k<j} x_{i,k}\left(x^{a-1}\right)_{k,j}
\,.
\]
We claim that every term in the sum is zero. Indeed if $\alpha(i)\not\in
[i,j]$, $p_\alpha(i+1,j)=p_\alpha(i,j)-1$ so that $a-1\ge p_\alpha(i,j)-1= p_\alpha(i+1,j)\ge p_\alpha(k,j)$ and we apply the induction to $\left(x^{a-1}\right)_{k,j}$.
Similarly if $\alpha(i)=j$ and $a>p_\alpha(i,j)$, $a-1\ge p_\alpha(i,j)=p_\alpha(i+1,j)\ge p_\alpha(k,j)$.

So we can assume in what follows that $\alpha(i)\in ]i,j[$. The exact same
reasoning applied to $j$ allows to conclude that $\alpha(j)\in ]i,j[$, so that
$i<\alpha(i)<\alpha(j)<j$.

We now come to the crucial remark that 
\begin{align*}
p_\alpha(i,j)&=j-i+1-\#\{\text{pairings of $\alpha$ inside $[i,j]$}\}\\
&=j-\alpha(j)+1-\#\{\text{pairings of $\alpha$ inside $[\alpha(j),j]$}\}\\
&+\alpha(j)-\alpha(i)-1-\#\{\text{pairings of $\alpha$ inside $[\alpha(i)+1,\alpha(j)-1]$}\}\\
&+\alpha(i)-i+1-\#\{\text{pairings of $\alpha$ inside $[i,\alpha(i)]$}\}\\
&=p_\alpha(i,\alpha(i))+p_\alpha(\alpha(i)+1,\alpha(j)-1)+p_\alpha(\alpha(j),j)
\end{align*}
where in the last line, if $\alpha(i)+1=\alpha(j)$ then conventionally
$p_\alpha(\alpha(i)+1,\alpha(j)-1)=0$.
Indeed the configuration does not allow for mixed pairings
between the three intervals $[i,\alpha(i)]$, $[\alpha(i)+1,\alpha(j)-1]$, $[\alpha(j),j]$.
So we can write
\[
\left(x^a\right)_{i,j}=
\sum_{i<k<\ell<j} \left(x^{p_\alpha(i,\alpha(i))}\right)_{i,k}
\left(x^{a-p_\alpha(i,\alpha(i))-p_\alpha(\alpha(j),j)}\right)_{k,\ell}
\left(x^{p_\alpha(\alpha(j),j)}\right)_{\ell,j}
\,.
\]
The first factor is zero if $k\le \alpha(i)$ by the
induction hypothesis, noting that $p_\alpha(i,k)\le p_\alpha(i,\alpha(i))$
and similarly the third factor
is zero if $\ell\ge \alpha(j)$. If $\alpha(i)+1=\alpha(j)$ the proof
is finished; otherwise note that $p_\alpha(k,\ell)\le \alpha_\alpha(\alpha(i)+1,\alpha(j)-1)
\le a-p_\alpha(i,\alpha(i))-p_\alpha(\alpha(j),j)$ and the second factor vanishes
for the same reason.
\end{proof}

Taking $a=N\ge p_\alpha(i,j)$ in Lemma~\ref{lem:poweq},
we find that $\hat\O_\alpha \subset \O$ (set-theoretically).

Now observe that $\hat\O_\alpha$ is defined by $p$ equations, but $\O$
is equidimensional of codimension $p$, so that $\hat\O_\alpha$ is 
(a complete intersection) of pure
codimension $p$, and so is a union of irreducible components of $\O$.

One could with more effort conclude geometrically
that $\O_\alpha=\hat\O_\alpha$, but instead we shall use multidegrees and
the uniqueness property of Lemma~\ref{lem:deg}.
Define $\hat J_\alpha:=\mdeg_\n \hat\O_\alpha$. Since the $\hat\O_\alpha$ are complete
intersections, one can calculate directly
\begin{equation}\label{Jtwocol}
\hat J_\alpha=\prod_{i<j=\alpha(i)} \left(\frac{j-i+1}{2}\A+z_i-z_j\right)
\,.
\end{equation}
Note that an identical formula (really, at $h=0$, but this can
be absorbed in the shift of the $z$'s) appears in \cite{KL-KL}
in an indirectly related context.

We can then check
\begin{lem}
$\hat J_\la=\sum_\alpha \hat J_\alpha e_\alpha$ satisfies Eq.~\eqref{exchbis}.
\end{lem}
\begin{proof}
Recall that Eq.~\eqref{exchbis} amounts to saying that the entries
$\hat J_\alpha$ of $\hat J_\la$ must satisfy Eq.~\eqref{exch}.
Taking into
account that we use the {\em dual}\/ action defined above,
we find the same dichotomy as in Section~\ref{sec:dicho}, but inverted:
\begin{enumerate}
\item If $\alpha(i)=i+1$, then according to Section~\ref{sec:dicho}
case (2), $m_{i;\alpha',\alpha'}=+1$, so $m_{i;\alpha,\alpha}=-1$, i.e.
we are in case (1) of Section~\ref{sec:hotta};
so Eq.~\eqref{exch} can be rewritten as Eq.~\eqref{excha},
which is trivially satisfied by
$\hat J_\alpha= C (h+z_i-z_{i+1})$ where $C$ does not depend on $z_i,z_{i+1}$.
\item If $\alpha(i)\ne i+1$, then according to Section~\ref{sec:dicho}
case (2), $m_{i;\alpha',\alpha'}=-1$, so $m_{i;\alpha,\alpha}=+1$, i.e.
we are in case (2) of Section~\ref{sec:hotta};
so Eq.~\eqref{exch} can be rewritten as Eq.~\eqref{exchb}, or more explicitly,
\[
-(h+z_i-z_{i+1})\der_i \hat J_\alpha=
\begin{cases}
0&\alpha(i)=\alpha(i+1)=\varnothing\\
\hat J_{\tl\alpha}&\text{otherwise}
\end{cases}
\]
where again we have identified standard Young tableaux
and link patterns when we write ``$\tl\alpha$''.

This equation can also be checked directly case by case:
\begin{itemize}
\item If $\alpha(i)=\alpha(i+1)=\varnothing$, $\hat J_\alpha$ does not
depend on $z_i, z_{i+1}$, so $\der_i \hat J_\alpha=0$.
\item If $\alpha(i)=\varnothing$, $\alpha(i+1)=j>i+1$,
$\hat J_\alpha=C (\frac{j-i}{2}h+z_{i+1}-z_j)$, so $-(h+z_i-z_{i+1})\der_i \hat J_\alpha
= C(h+z_{i+1}-z_i)$ which is indeed $\hat J_{\tl\alpha}$ since $\tl\alpha$
differs from $\alpha$ only in pairing $i,i+1$ and having $j$ unpaired.
\item The case $\alpha(i)=j<i$, $\alpha(i+1)=\varnothing$ can be treated
similarly.
\item If $\alpha(i)=j<i$, $\alpha(i+1)=k>i$,
$\hat J_\alpha=C(\frac{i-j+1}{2}h+z_j-z_i)(\frac{k-i}{2}h+z_{i+1}-z_k)$, so
$-(h+z_i-z_{i+1})\der_i \hat J_\alpha=C(h+z_i-z_{i+1})(\frac{k-j+1}{2}h+z_j-z_k)$
which again coincides with $\hat J_{\tl\alpha}$, since $\tl\alpha$ pairs
$i,i+1$ and $j,k$.
\item The other two cases $i<i+1<\alpha(i+1)<\alpha(i)$ and
$\alpha(i+1)<\alpha(i)<i<i+1$ can be treated similarly.
\end{itemize}
\end{enumerate}
\end{proof}

Applying the intertwiner $\phi$ and then Theorem~\ref{thm:ident} 
and Lemma~\ref{lem:deg},
we conclude that the $\hat J_\alpha$ coincide up to normalization
with the multidegrees $J_\alpha$ of
the orbital varieties: $\hat J_\alpha=c J_\alpha$ for some $c\ne0$.

Now according to the above, $\hat\O_\alpha$ is a union of a certain certain
subset of $\O_\beta$, 
so we can write at the level of multidegrees
$\hat J_\alpha=\mdeg_\n \hat\O_\alpha = \sum_\beta k_{\alpha,\beta} J_\beta$
where $k_{\alpha,\beta}$ is the multiplicity of $\O_\beta$ in $\hat\O_\alpha$
(or zero if $\O_\beta\not\subset\hat\O_\alpha$). 
In order to conclude,
we only need to note that according to Eq.~\eqref{exch} and 
Theorem~\ref{thm:hotta}, the $J_\beta$, $\beta\in\SYT$, generate a subspace
of $\mathbb{C}[z_1,\ldots,z_n,h]$ which is
an irreducible representation of the symmetric
group under the action $\hat s_i$ 
(associated to the conjugate partition $\lambda'$, with our sign
convention), and so in particular are linearly independent.
So $c J_\alpha=\sum_\beta k_{\alpha,\beta} J_\beta$ implies 
$k_{\alpha,\beta}=c\delta_{\alpha,\beta}$ and
$\hat\O_\alpha=\O_\alpha$ (set-theoretically).
In other words, we have proved:
\begin{thm}\label{thm:defeqb}
$\O_\alpha$ is defined by the equations
\[
\left(x^{\frac{j-i+1}{2}}\right)_{i,j}=0,\qquad i<j=\alpha(i)
\,.
\]
\end{thm}
In fact,
since all coefficients of $\hat J_\alpha=k_{\alpha,\alpha}J_\alpha$ 
at $h=0$ are $\pm 1$ and 
$J_\alpha\in\mathbb{Z}[z_1,\ldots,z_n,h]$, we have $k_{\alpha,\alpha}=c=1$,
i.e., $\hat J_\la=J_\la$
(and $\hat\O_\alpha$ being a complete intersection,
the equations above define $\O_\alpha$ as a reduced scheme).
In particular,
\begin{equation}\label{Jtwocolb}
J_\alpha=\prod_{i<j=\alpha(i)} \left(\frac{j-i+1}{2}\A+z_i-z_j\right)
\,.
\end{equation}

\subsubsection{Change of basis}
We can use again duality (conjugation of partition and Young tableaux)
to find the intertwiner:
\begin{lem}
The intertwiner is:
\[
\phi(e_\alpha)=(-1)^{p(p-1)/2}
\,\epsilon_\alpha\sum_{\substack{L=(l_1,\ldots,l_n)\\{\scriptscriptstyle\text{ permutation of }
(1,\ldots,p,1,\ldots,n-p)}}} J_{\alpha'}|_{z_1^{l_1-1}\cdots z_n^{l_n-1}}
\,v_L
\]
where $|_{z_1^{l_1-1}\cdots z_n^{l_n-1}}$ denotes the given coefficient of
a polynomial.
\end{lem}
\begin{proof}
We need to check that this $\phi$ intertwines the symmetric group action.
Note that $\deg J_{\alpha'}=p(p-1)/2+(n-p)(n-p-1)/2$ 
so $z_1^{l_1}\cdots z_n^{l_n}$ exhausts
the degree and we might as well set $h=0$ in $J_{\alpha'}$.
\begin{align*}
P^{(i,i+1)}\phi(e_\beta)
&=
(-1)^{p(p-1)/2}\epsilon_\beta
\sum_L J_{\beta'}|_{z_1^{l_1-1}\cdots z_n^{l_n-1}}
\,v_{l_1,\ldots,l_{i+1},l_i,\ldots,l_n}
\\
&=
(-1)^{p(p-1)/2}\epsilon_\beta
\sum_L (\tau_i J_{\beta'})|_{z_1^{l_1-1}\cdots z_n^{l_n-1}}
\,v_L
\\
&=
(-1)^{p(p-1)/2}\epsilon_\beta
\sum_L \sum_{\alpha'}m_{i;\beta',\alpha'} J_{\alpha'}|_{z_1^{l_1-1}\cdots z_n^{l_n-1}}
\,v_L
&&\text{by Eq.~\eqref{exch} at $h=0$}
\\
&=\sum_{\alpha'}\epsilon_\beta\epsilon_\alpha
 m_{i;\beta',\alpha'} \phi(e_\alpha)
\\
&=-\sum_\alpha 
m_{i;\alpha,\beta} \phi(e_\alpha)&&\text{by Eq.~\eqref{eq:dualaction}}
\\
&=\phi(\rho^{(i,i+1)}e_\beta)&&\text{by Eq.~\eqref{exchdual}}
\end{align*}

Next we check the normalization, which is fixed by Theorem~\ref{thm:ident}. 
Consider the same tableau $\alpha_0$ that
was used in the proof of Proposition~\ref{prop:intertworow},
i.e.,
$\setlength{\cellsize}{18pt}\def\boxformat{\scriptstyle}
\alpha_0=\tableau{1&&\cdots&&n-p\\\vbox{\hbox{$\scriptstyle n-p$}\hbox{$\scriptstyle+1$}}&\cdots&n}$, which is a tableau of $\la'$.
Then we have as before
$J_{\alpha_0}=\prod_{1\le i<j\le n-p}(h+z_i-z_j)\prod_{n-p+1\le i<j\le n} (h+z_i-z_j)$,
so that $\phi(e_{\alpha'_0})=(-1)^{(n-p)(n-p-1)/2} A_{n-p}\otimes A_n$
where $A_k$ is the antisymmetrizer $\sum_{\sigma\in S_k}
(-1)^\sigma v_{\sigma(1)}\otimes\cdots\otimes v_{\sigma(k)}$.

Now consider the multi-index $L=(n-p,\ldots,1,1,\ldots,p)$. According
to Eq.~\eqref{Ihzero}, $I_L=(z_n-z_{n-2p+1})(z_{n-1}-z_{n-2p+2})\cdots (z_{n-p+1}-z_{n-p})$.
But this is also $J_{\alpha'_0}|_{h=0}$ according to Eq.~\eqref{Jtwocolb}.
Recalling that Theorem \ref{thm:hotta} implies that the
$J_\alpha|_{h=0}$ are linearly independent, we conclude that
the coefficient of $v_L$ in $\phi(e_{\alpha'_0})$ must be $1$, which is
consistent with the formula above.
\end{proof}

\subsubsection{Cyclicity}
In order to simplify the discussion, we assume now that $n=2p=2N$,
i.e., the Young diagram is rectangular,
and use again the rotation $\rho$ of Section~\ref{sec:cycl}
(in principle the content
of the present section is valid for $p<n/2$ but more work would be
needed to define $\rho$).

Then it is obvious from the explicit form \eqref{Jtwocolb} that
the $J_\alpha$ satisfy the extra ``cyclicity'' relation
\[
J_{\rho\alpha}(z_1,\ldots,z_n)=-J_\alpha(z_2,\ldots,z_n,z_1-(N+1)h)
\]
Indeed if $n$ is paired to say $i$, there is a factor
$\frac{n-i+1}{2}h+z_n-z_i$, but once rotated, the pairing $1,i+1$
produces a factor $\frac{i+1}{2}h+z_{i+1}-z_1$
which corresponds to the substitution $z_i\to z_{i+1}$, 
$z_n\to z_1-\frac{n+2}{2}h$ and
a change of sign. 

This together with Eq.~\eqref{exch} implies the \qKZ/ equation.

We shall see in next section the condition
on $\lambda$ for $I_\la$ (and therefore $J_\la$) to satisfy \qKZ/, a condition
which is satisfied if $\lambda$ has two columns.


\section{$I_\la$ is a \qi-conformal block,
$I_\la$ satisfies the \qKZ/ equations if $d(\la)\leq 1$.}
\label{sec:proofs}

\begin{thm}
\label{thm:I_is_qCB}
Let $\la\in\N^N$ be a partition of $n$. If $d(\la)>0$ then
$I_\la$ is a \qi-conformal block at level $d(\la)$. If
$d(\la)=0$ then $I_\la$ is a \qi-conformal block at level 1.
\end{thm}

\begin{proof} First we prove that $I_\la$ is a singular
vector. Observe that
the $e_{\ij}$-image of a $V[\la]$-valued function
is a $V[\mu]$-valued function with $\mu_k=\la_k$
except $\mu_i=\la_i+1$ and $\mu_j=\la_j-1$.
The action of $e_{\ij}$ and the deformed action of $S_n$ commute,
hence $e_{\ij}I_\la$ is skew-symmetric. If $i<j$, then the degree
$k(\la)$ of $e_{i,j}I_\la$ is strictly less than $k(\mu)=k(\la)+\la_i-\la_j+1$.
Hence $e_{\ij}I_\la$ must be 0 by Lemma \ref{lem:deg}.

The operation $e(z)$ also commutes with the deformed action of $S_n$,
hence we can argue for the \qi-conformal block property similarly.

First let $d(\la)>0$. The function $e(z) I_\la$ is a
skew-symmetric $V[\mu]$-valued function of degree
$k(\la)+1$, where $\mu_k=\la_k$ except
$\mu_1=\la_1+1$, $\mu_N=\la_N-1$.
Calculation shows that $k(\mu)=k(\la)+d(\la)+1$
which is strictly greater than the degree $k(\la)+1$.
Hence $e(z) I_\la=0$ by Lemma \ref{lem:deg}.

Now let $d(\la)=0$. The function $e(z)^2 I_\la$
is a skew-symmetric $V[\mu]$-valued function of
degree $k(\la)+2$, where $\mu_k=\la_k$
except $\mu_1=\la_1+2$, $\mu_N=\la_N-2$.
Calculation shows that $k(\mu)=k(\la)+4$ which
is strictly greater than the degree $k(\la)+2$.
Hence $e(z)^2 I_\la=0$ by Lemma \ref{lem:deg}.
\end{proof}

\begin{cor} \label{cor:dim}
Let $d(\la)\le 1$. For generic $z\in \C^n$ the space of \qi-conformal blocks
at level 1 is one-dimensional.
\end{cor}

\begin{proof}
We recalled in Lemma \ref{lem:dim} that for generic $z$ the dimension of \qi-conformal blocks is at most 1.
We proved in Theorem \ref{thm:I_is_qCB} that $I_\la$ is a (generically nonzero) \qi-conformal block. Hence, for generic $z$ this space is one-dimensional.
\end{proof}

Consider $z\in \C^n$ for which Corollary \ref{cor:dim} holds and for which the \qKZ/
operators have no singularities, e.g.~$z_1-z_2+h\not=0$ etc. Over the configuration space of these $z$'s one
may consider the bundle of singular vectors. The \qi-conformal blocks form a rank 1 subbundle.
It is proved in \cite{MV1} that the subbundle of \qi-conformal blocks is preserved by the
\qKZ/ connection. In our language this means the following theorem.

\begin{thm} \cite[Theorem 2]{MV1} \label{thm:qCB-qKZ}
If, for the $z$'s defined above, the space of \qi-conformal blocks at level 1 is spanned by a
$V[\la]$-valued function $I$, then a scalar function multiple of $I$ satisfies the \qKZ/ difference
equations~\eqref{eq:qkz}.
\end{thm}

\begin{thm} \label{thm:I_is_qKZ}
Let $d(\la)\le 1$.
Then $I_\la$ satisfies the \qKZ/ equations \eqref{eq:qkz}.
\end{thm}

The rest of this section is the proof of this theorem.

\begin{proof}
The function $I_\la$ is a \qi-conformal block at level 1.
By Theorem \ref{thm:qCB-qKZ}, there is a scalar function
$f(z_1,\ldots,z_n)$ such that $fI_\la$
satisfies the \qKZ/ equations.
We obtain
\be
f(\ldots,z_i-(N+1)h,\ldots)\>I_\la(\ldots,z_i-(N+1)h,\ldots)\,=
\,K_i f I_\la,\qquad i=1,\ldots,n.
\ee
After rearrangement we have
\beq
\label{eq:arranged}
\frac{f(\ldots,z_i-(N+1)h,\ldots)}{f}
I_\la(\ldots,z_i-(N+1)h,\ldots)=K_iI_\la,\qquad i=1,\ldots,n.
\eeq
Since $I_\la(\ldots,z_i-(N+1)h,\ldots)$
and $K_i I_\la$ are rational functions of $z_1,\ldots,z_n$, the ratios
\be
g_i =\frac{f(\ldots,z_i-(N+1)h,\ldots)}{f},\qquad i=1,\dots,n,
\ee
are rational functions of $z_1,\ldots,z_n$ of degree 0.

Our first goal is to show that $g_1$ is a constant function equal to \,1\,.
We start with two lemmas.

\begin{lem} \label{lem:skew}
Let $I$ be a $V[\la]$-valued skew-symmetric function (for example $I=I_\la$).
Then
\be
R^{(\ii+1)}(z_i-z_{i+1})\>I\,=\,-\>P^{(\ii+1)}\>I(z_i\lrar z_{i+1})\,.
\ee
\end{lem}

\begin{proof}
Skew symmetry with respect to the transposition $s_i$ implies the formula by
direct calculation.
\end{proof}

\begin{lem}
\label{lem:pol}
$K_1 I_\la $ is a polynomial.
\end{lem}

\begin{proof}
By Lemma~\ref{lem:skew},
\;$K_1 I_\la \>=\>
(-1)^{n-1}\>P^{(1,n)}\cdots P^{(1,2)}\,I_\la(z_2\lc z_n,z_1)$\,.
\end{proof}

We claim that the components of $I_\la$ do not have
a common polynomial factor of degree~$\ge 1$. Indeed,
if a polynomial, necessarily homogeneous, divides
all components of $I_\la$, then in the
quasiclassical limit $h=0$ a polynomial would
divide all components of $I_{\la}^{h=0}$. One sees
from the explicit form of $I_{\la}^{h=0}$ in the Remark in Section \ref{sec:Ila}
that this is not the case.

This claim together with Lemma \ref{lem:pol}
implies that the denominator of $g_1$ is a constant
function, and since $g_1$ is a rational function
of degree 0, $g_1$ must be a constant function.
The $h=0$ limit of equation \eqref{eq:arranged}
then implies that $g_1=1$. In other words, we proved that $I_\la$ satisfies the
first \qKZ/ equation.

\smallskip
Our next goal is to show that the first \qKZ/ equation
implies the others for skew-symmetric functions.

For brevity we will write $p$ for $(N+1)h$. Assume that the $i$-th \qKZ/ equation
\be
I(z_i\to z_i-p)=R^{(\ii-1)}(z_i-z_{i-1}-p)\cdots R^{(i,1)}(z_i-z_1-p) R^{(i,n)}(z_i-z_n)\cdots R^{(\ii+1)}(z_i-z_{i+1}) I
\ee
holds. At the right end of the formula we can use Lemma \ref{lem:skew} to obtain
\be
I(z_i\to z_i-p)=R^{(\ii-1)}(z_i-z_{i-1}-p)\cdots R^{(i,1)}(z_i-z_1-p) \times \qquad \qquad\qquad
\ee
\be
\qquad\qquad\qquad\qquad \times R^{(i,n)}(z_i-z_n)\cdots R^{(\ii+2)}(z_i-z_{i+2}) (-P^{(\ii+1)}I(z_i \lrar z_{i+1})).
\ee
Applying $-P^{(\ii+1)}$ to this equation, together with the iterated application of
\be
P^{(\ii+1)} R^{(i,m)}(u)=R^{(i+1,m)}(u)P^{(\ii+1)}
\ee
we get
\be
-P^{(\ii+1)}I(z_i\to z_i-p)=R^{(i+1,i-1)}(z_i-z_{i-1}-p)\cdots R^{(i+1,1)}(z_i-z_1-p) \times \qquad \qquad\qquad
\ee
\be
\qquad\qquad\qquad\qquad \times R^{(i+1,n)}(z_i-z_n)\cdots R^{(i+1,i+2)}(z_i-z_{i+2}) I(z_i \lrar z_{i+1}).
\ee
To this equation we substitute $z_i\lrar z_{i+1}$ and obtain
\begin{align*}
-P^{(\ii+1)}\,I(z_i\to z_{i+1}-p, z_{i+1}\to z_i)\,={}&\,
R^{(i+1,i-1)}(z_{i+1}-z_{i-1}-p)\cdots R^{(i+1,1)}(z_{i+1}-z_1-p) 
\\
&{}\times\,R^{(i+1,n)}(z_{i+1}-z_n)\cdots R^{(i+1,i+2)}(z_{i+1}-z_{i+2})\,I\;.
\end{align*}
We can use Lemma \ref{lem:skew} to write the left hand side in the form of
\be
R^{(\ii+1)}(z_i-z_{i+1}+p)I(z_{i+1}\to z_{i+1}-p).
\ee
Then applying the $R^{(\ii+1)}(z_{i+1}-z_i-p)$ operator to both sides
results in the $i+1$-st \qKZ/ equation. This finishes the proof of
Theorem~\ref{thm:I_is_qKZ}.
\end{proof}

\section{A \qi-Selberg type integral}
\label{sec:Selberg}

According to the general principle in \cite{MV3}, if a \KZ/-type equation
has a one-dimensional space of solutions, then the hypergeometric or
\qi-hypergeometric integrals representing the solutions can be calculated
explicitly, see demonstrations of that principle in
\cite{FStV, RSV, TV1, TV2, V, W1, W2}.
In this section we give another example of this type.

In the rest of the paper we fix $N=2$ and consider the bundle of
the $\glt$ \qi-conformal blocks at level 1 over $(\C^2)^{\otimes n}$.
That bundle is of rank 1.
The \qKZ/ operators define a discrete flat connection on that bundle. In
Section~\ref{sec:Ila} we constructed a vector-valued polynomial on $\C^n$ that
generates the space of flat sections of the discrete connection. Moreover,
in Section \ref{sec:geom} we identified that polynomial with the generating
function of extended Joseph polynomials of orbital varieties associated with
nilpotent $n\times n$-matrices. On the other hand in \cite{TV1, MV1, MV2}
flat sections of the same connection were constructed as multidimensional
\qi-hypergeometric integrals. In this section we identify the
\qi-hypergeometric flat sections constructed in \cite{TV1, MV1, MV2} with the
polynomial section constructed in Section \ref{sec:Ila} and obtain a
\qi-Selberg type identity that a multidimensional \qi-hypergeometric integral
equals a polynomial.

\subsection{Quantized conformal blocks at level 1 and \qKZ/ equations}

\label{sec:def}
\label{sec:connection}

First we recall some earlier definitions specialized for $\glt$, and with the substitution $h=1$.
We consider the
vector representation $\C^2$ of $\glt$ with the standard basis $v_1,v_2$ and denote $V=(\C^2)^{\otimes n}$.
The space $V$ has a basis
of vectors $v_{i_1}\otimes\dots\otimes v_{i_n}$, where
\;$i_j\in\{1,2\}$. Every such a sequence $(i_1,\dots,i_n)$ defines
a decomposition $L=(L_1,L_2)$ of $\{1,\dots,n\}$ into disjoint subsets,
\;$L_j=\{l\ |\ i_l=j\}$. The basis vector
$v_{i_1}\otimes\dots\otimes v_{i_n}$ is denoted by $v_L$.
We have $V=\bigoplus_{\la=(\la_1,\la_2)} V[\la],$
where $\la_1+\la_2=n$, and
$V[\la]=\{v\in V\ |\ e_{i,i}v=\la_i v , i=1,2\}$.
Denote by $\Il$ the set of all indices $L$ with $|L_j|=\la_j$, \;$j=1,2$.
The vectors $\{v_L\ |\ L\in\Il\}$ form a basis of $V[\la]$.

\medskip
For $z=(z_1,\dots,z_n)\in\C^n$, we define an operator
$e(z) : V\to V$, by the formula
\[
e(z) = \sum_{j=1}^n (z_j - e_{2,2}^{(j)} + \sum_{s=j+1}^n
(e_{1,1}-e_{2,2})^{(s)})e_{1,2}^{(j)}.
\]
For $\la=(\la_1,\la_2)$ with $1\geq \la_1-\la_2\geq 0$, we define the { space
of quantized conformal blocks at level 1} as
\[
CB_\la(z) =\{ v\in V[\la]\ | \ e_{1,2}v=0,\ {} e(z)^{2+\la_2-\la_1}v=0\,\}.
\]


Note that for any $n$, the space $V$ has a unique subspace $V[\la]$ with $1\geq \la_1-\la_2\geq 0$.
If $n=2\ell$, then $\la=(\ell,\ell)$ and if $n=2\ell+1$, then $\la=(\ell+1,\ell)$. That $\la$ will be called
the {\it middle weight}. The middle weight $\la$ is determined by $n$ and
we will denote the subspace $CB_\la(z)$ just by $CB(z)$.

Corollary \ref{cor:dim} claims that for generic $z\in\C^n$, dim $CB(z)=1$.



For $i=1,\dots,n$, the \qKZ/ operators at level 1 on $V$ are
\begin{multline*}
K_i(z_1,\dots,z_n)
=
R^{(i,i-1)}\left(z_i-z_{i-1}-3\right)\cdots
R^{(i,1)}\left(z_i-z_1-3\right)\times
\\
\times R^{(i,n)}(z_i-z_n)\cdots R^{(i,i+1)}(z_i-z_{i+1}).
\end{multline*}
The \qKZ/ operators define on $V$ a discrete flat connection,
\[
K_j(z_1,\dots,z_i-3,\dots,z_n) K_i(z_1,\dots,z_n)
=
K_i(z_1,\dots,z_j-3,\dots,z_n) K_j(z_1,\dots,z_n)
\]
for all $i,j$, see \cite{FR-qKZ}.
A $V$-valued function $I(z)$ is a {\em flat section} if
it satisfies the \qKZ/ equations,
\begin{equation}
\label{qkz}
\phantom{aaaa}
I(z_1,\ldots, z_i-3,\ldots,z_n)=K_i(z_1,\ldots,z_n)I(z_1,\ldots,z_i,\ldots,z_n),
\qquad i=1,\ldots,n.
\end{equation}
The subbundle of conformal blocks at level 1 is invariant with respect
to the \qKZ/ connection,
\[
K_i(z_1,\ldots,z_n) : CB(z_1,\ldots,z_n) \to CB(z_1,\ldots,z_i-3,\ldots,z_n)
\]
for all $i$, see \cite{MV1, MV2}.











\bigskip

Recall the polynomial $I_\la$ from Section \ref{sec:Ila}, with the substitution $h=1$.
In notation we will not indicate the $h=1$ substitution. Hence in the rest of the paper we have e.g.
\[
I_{(2,1)}=(z_1-z_2+1)v_{112}+(z_3-z_1-2)v_{121}+(z_2-z_3+1)v_{211}.
\]

Results of the first part of the paper, in our present conventions, are as follows.
\begin{itemize}
\item{} $I_\la$ is skew symmetric with respect to the $S_n$-action
\eqref{Sn} with $h=1$.
\item{} $I_\la$ has degree $k(\la)=\frac 12\la_1(\la_1-1)+\frac 12\la_2(\la_2-1)$ (the minimal degree skew symmetric polynomial).
\item{} Given $\la$, let $L=(L_1,L_2)$ be the partition of $\{1,\dots,n\}$ with
$L_1=\{1,\dots,\la_1\}$. Then $I_\la$ is normalized in such a way that its $L$-th coordinate is
\[
\prod_{1\leq a<b\leq\la_1}(z_a-z_b+1)
\prod_{\la_1< a<b\leq n}(z_a-z_b+1).
\]
That polynomial $I_\la$ will be called {\it minimal}.
\item{} If $\la=(\la_1,\la_2)$ is the middle weight, i.e. $1\geq \la_1-\la_2\geq 0$,
then $I_\la\in CB(z)$ and
$I_\la$ satisfies the \qKZ/ equations \eqref{qkz}.
\end{itemize}

\subsection{An integral representation for quantized conformal blocks
at level 1}
\label{sec:int repn}

In this section $\la$ is the middle weight for $n$, $\la=(\ell,\ell)$ if $n=2\ell$
and $\la=(\ell+1,\ell)$ if $n=2\ell+1$.


Define the {\it master function}
\begin{multline*}
\Phi(t_1,\dots,t_\ell, z_1,\dots,z_n) =
\prod_{n\geq j>i\geq 1}
\frac
{\Gamma((z_j-z_i+1)/3)}{\Gamma((z_j-z_i-1)/3)}
\prod_{\ell\geq j>i\geq 1}
\frac
{\Gamma((t_j-t_i+1)/3)}{\Gamma((t_j-t_i-1)/3)}
\times
\\
\times
\prod_{i=1}^n\prod_{j=1}^\ell
\frac
{\Gamma((z_i-t_j-1)/3)}{\Gamma((z_i-t_j)/3)} .
\end{multline*}

For $L=(L_1,L_2)\in \Il$ with $L_2=\{i_1<\dots <i_\ell\}$,
define the {function} $w_{L}(t_1,\dots,t_\ell, z_1,\dots,z_n)$ by the formula
\[
w_L = \sum_{\sigma\in S_\ell}
\prod_{j=1}^\ell \frac{1}{t_{\sigma_j}-z_{i_j}}
\prod_{m=1}^{i_j-1} \frac{t_{\sigma_j}-z_m+1}{t_{\sigma_j}-z_m}
\prod_{1\le i<j\le \ell, \ \sigma_i>\sigma_j}
\frac{t_{\sigma_i}-t_{\sigma_j}+1}{t_{\sigma_i}-t_{\sigma_j}-1}
\]
Define the $V[\la]$-valued {\it weight function} by the formula
\[
w(t_1,\dots,t_\ell,z_1,\dots,z_n) = \sum_{L\in\Il} w_L(t_1,\dots,t_\ell,z_1,\dots,z_n)v_L .
\]


Define the {\it trigonometric weight function}
$W(t_1,\dots,t_\ell, z_1,\dots,z_n)$ by the formula
\[
\label{W}
W = \pi^{\ell}\prod_{j=1}^{\ell}
\frac {\sin(\pi(z_{2j}-z_{2j-1}+1)/3)} {\sin(\pi(t_j-z_{2j-1})/3)\sin(\pi(t_j-z_{2j})/3)}
\prod_{m=1}^{2j-2}\frac{\sin(\pi(t_j-z_m+1)/3)}{\sin(\pi(t_j-z_m)/3)}.
\]

Using the formula $\Gamma(1-x) \; \Gamma(x) = \frac{\pi}{\sin{(\pi x)}}$, we can write
\begin{multline*}
\Phi W = \prod_{n\geq j>i\geq 1}
\frac
{\Gamma((z_j-z_i+1)/3)}{\Gamma((z_j-z_i-1)/3)}
\prod_{\ell\geq j>i\geq 1}
\frac
{\Gamma((t_j-t_i+1)/3)}{\Gamma((t_j-t_i-1)/3)}
\times
\\
\times \prod_{j=1}^\ell \prod_{m=0}^1\Gamma((z_{2j-m}-t_j-1)/3)\Gamma(1-(z_{2j-m}-t_j)/3)
\times
\\
\phantom{aaaaaa}
\times
\prod_{j=1}^\ell \prod_{i=1}^{2j-2}
\frac
{\Gamma(1-(z_i-t_j)/3)}{\Gamma(1-(z_i-t_j-1)/3)}
\prod_{i=2j+1}^n
\frac
{\Gamma((z_i-t_j-1)/3)}{\Gamma((z_i-t_j)/3)} \times
\\
\times \pi^{-\ell}\prod_{j=1}^\ell \sin(\pi(z_{2j}-z_{2j-1}+1)/3) .
\end{multline*}
The function $\Phi W w$ is a meromorphic function of $t_1,\dots,t_\ell$ with first order poles at the
hyperplanes
\begin{align*}
t_i - t_j &= 1 + 3 s ,
&&
i<j , \ {} \ s = 1,2,...
\\
t_j - z_m &= -1 + 3 s ,
&&
m \le 2j ,\ {} \ {} s = 0,1,...
\\
t_j - z_m &= - 3 s ,
&&
m \ge 2j-1 , \ {} \ s = 0,1,...
\end{align*}

Define an oriented unbounded one-chain $\mc C_n\subset \C$. It consists of
the vertical line $-\frac 12 + \sqrt{-1}\R$, oriented from
$-\frac 12 -\sqrt{-1}\infty$ to $-\frac 12+\sqrt{-1}\infty$,
and $2n$ circles $C_1, \dots, C_{2n}$ of radius $\frac 14$.
The circle $C_j$ for $1\le j\le n$ is centered at $j\sqrt{-1}$ and
oriented counterclockwise, while the circle $C_j$ for $n<j\le 2n$
is centered at $-1+j\sqrt{-1}$ and oriented clockwise.
Define the {\it integration cycle}
\[
\mc C_n^\ell = \{(t_1,\dots,t_\ell)\in \C^\ell\ | \ t_i\in \mc C_n\} .
\]


For $(z_1,\dots,z_n)\in \C^n$ such that $|z_j-j\sqrt{-1}|< \frac 14$,
define a $V[\la]$-valued {\it \qi-hypergeometric integral} by the formula
\[
\label{int}
\Psi_\la(z_1,\dots,z_n)=\int_{\mc C_n^\ell }
\Phi(t,z)w(t,z)W(t,z)\,dt_1\ldots dt_\ell .
\]

\begin{thm} [\cite{TV1}]
\label{thm analyt cont }
The function $\Psi_\la(z_1,\dots,z_n)$ is well-defined and extends to a meromorphic
function on $\C^n$. Moreover, the function $\Psi_\la(z_1,\dots,z_n)$ is a solution of the \qKZ/
equations \eqref{qkz}.
\end{thm}

\begin{thm} [\cite{MV1}]
\label{thm int CB}
For generic $z\in \C^n$, we have $\Psi_\la(z_1,\dots,z_n) \in CB(z)$.
\end{thm}

\begin{thm}
\label{thm main}
Let $n=2\ell,$ $ \la=(\ell,\ell)$ or $n=2\ell+1$, $\la=(\ell+1,\ell)$.
Let $I_\la(z_1,\dots,z_n)$ be the minimal skew-symmetric $V[\la]$-valued polynomial as above.
Then the \qi-hypergeometric integral $\Psi_\la(z_1,\dots,z_n)$ equals the polynomial $c_n I_\la(z_1,\dots,z_n)$, where
\begin{equation}
\label{c}
c_2 = 2 \pi \sqrt{-1} \frac{\Gamma(2/3)\Gamma(-1/3)}{\Gamma(1/3)} ,
\qquad
c_{2\ell}= 3^{-\ell(\ell-1)} c_2^\ell,
\qquad
c_{2\ell+1} = (-1)^\ell 3^{-\ell^2} c_2^\ell.
\end{equation}
\end{thm}

By Section \ref{sec:geom} the polynomial $I_\la$ is the generating function of the extended
Joseph polynomials of the orbital varieties
associated with nilpotent $n\times n$-matrices. Hence, Theorem \ref{thm main} gives
an integral representation for those extended Joseph polynomials.

\begin{proof}
For $n=2$ the middle weight is $(1,1)$ and we have
$I_{(1,1)}=v_{12}-v_{21}$, $\Psi_{(1,1)}=c_2I_{(1,1)}$, where
\begin{align*}
c_2&= \pi^{-1}\sin(\pi(z_{2j}-z_{2j-1}+1)/3) \frac {\Gamma((z_2-z_1+1)/3)}{\Gamma((z_2-z_1-1)/3)}\times
\\
&\qquad\times
\int_{\mc C_2}
\prod_{m=0}^1\Gamma((z_{2j-m}-t_j-1)/3)\Gamma(1-(z_{2j-m}-t_j)/3) \frac {dt}{t-z_1}\\
&= 2 \pi \sqrt{-1} \frac{\Gamma(2/3)\Gamma(-1/3)}{\Gamma(1/3)} ,
\end{align*}
see \eqref{barnes}.

For arbitrary $n$ we have $\Psi_\la(z_1,\dots,z_n)=c_n(z_1,\dots,z_n) I_\la(z_1,\dots,z_n)$, where
$c_n(z_1,\dots,z_n)$ is a scalar function $3$-periodic with respect to every variable. Indeed, both
$\Psi_\la(z_1,\dots,z_n)$ and $I_\la(z_1,\dots,z_n)$ are quantized conformal blocks at level 1 and both
satisfy the \qKZ/ equations. To check that $c_n$ is given by \eqref{c}
we consider
the asymptotic zone:
\begin{enumerate}
\item[(i)]
$|z_{2i}-z_{2i-1}|\leq 1$, $|\text{Im}(z_{2i})|\leq 1$
for $i=1,\dots,\ell$,
\item[(ii)] Re$(z_{2i+2}-z_{2i})\gg 1$ for $i=1,\dots,\ell-1$ and
Re$(z_{n}-z_{n-2})\gg 1$ if $n$ is odd,
\end{enumerate}
use the Stirling formula for the Gamma functions,
\be
\frac{\Gamma((x+\al)/p)}{\Gamma((x+\beta)/p)} = (x/p)^{\al-\beta}(1+o(1)),
\qquad |\text{arg}(x/p)|<\pi,
\ee
and similarly to the proof of Theorem 6.7 in \cite{TV1} observe that
\[
c_{2\ell}= 3^{-\ell(\ell-1)} c_2^\ell (1+o(1)),
\qquad
c_{2\ell+1} = (-1)^\ell 3^{-\ell^2} c_2^\ell (1+o(1)).
\]
This proves the theorem.
\end{proof}

\section{An alternative integral for $N=2$}
\label{sec:altint}

Finally, we formulate an integral representation in
the two-row case $\lambda=(n-p,p)$,
distinct from that of the previous paragraph,
and which generalizes that of
\cite{artic42} (which was the case $n$ odd, $n=2p+1$).

Expanding $I_\la=\sum_L I_L v_L$, the multi-indices that contribute
to the sum have $n-p$ $1$s and $p$ $2$s. Let us parameterize
them as follows: denote $L(a)$ the multi-index
whose $2$s are located at indices $a=(a_1<\cdots<a_p)$.
\begin{thm}
\[
I_{L(a)}=(-1)^{p(n-p+1)}h^p\! \prod_{1\le i<j\le n }\! (h+z_i-z_j)
\oint \prod_{k=1}^p \frac{dw_k}{2\pi \sqrt{-1}}
\frac{\prod_{1\le k<\ell\le p} (w_\ell-w_k)(h+w_k-w_\ell)}%
{\prod_{k=1}^p\left(\prod_{i=1}^{a_k} (w_k-z_i)
\prod_{i=a_k}^{n} (h+w_k-z_i)\right)}
\]
The integration cycle is the product of $p$ identical 1-dimensional cycles. 
The 1-dimensional cycle is
any contour that surrounds once counterclockwise each of the $z_1,\ldots,z_n$ but none of the $z_1-h,\ldots,z_n-h$.
\end{thm}
Note that the integrals have no pole at infinity
(the integrand behaves as $w_k^{2(p-1)-(n+1)}$ as $w_k\to\infty$)
so we may as well consider that the contour
surrounds clockwise the $z_1-h,\ldots,z_n-h$ but none of the $z_1,\ldots,z_n$.
\begin{proof}
We are going to apply Lemma~\ref{lem:deg}. Denote by $\hat I_{L(a)}$
the r.h.s.\ of the formula above.

First one needs to check that $\hat I_{L(a)}$ is a polynomial
in $z_1,\ldots,z_n,h$. This is a routine calculation
based on the application of the residue formula for the $w_k$ integrals
and the check that would-be poles in the variables $z_i$ have vanishing
residue (see a similar calculation in \cite{artic45});
since the formula is homogeneous in $z_1,\ldots,z_n,h$ this leaves
only a power of $h$ in the denominator which is cancelled by the factor $h^p$.
The degree of $\hat I_{L(a)}$ is then (as a homogeneous polynomial
in $z_1,\ldots,z_n,h$)
$p+n(n-1)/2+p+2p(p-1)/2-p(n+1)=p(p-1)/2+(n-p)(n-p-1)/2=k(\la)$.

Next, we check that $\hat I_\la=\sum_L \hat I_L v_L$
is skew-symmetric by use of Lemma~\ref{fL}.
Fixing $i=1,\ldots,n-1$, there are four possibilities:
\begin{itemize}
\item If $L_i=L_{i+1}=1$, the integrand (including the prefactor in front of the integral) is $h+z_i-z_{i+1}$ times a symmetric function of $z_i, z_{i+1}$.
This implies that $\hat s_i I_L=-I_L$.
\item If $L_i=L_{i+1}=2$, say $a_k=i$, $a_{k+1}=i+1$,
then the integrand minus itself
with $z_i\leftrightarrow z_{i+1}$ is skew-symmetric in
$w_k, w_{k+1}$ and therefore its integral is zero.
This implies again that $\hat s_i I_L=-I_L$.
\item If $L_i=2$, $L_{i+1}=1$, say $a_k=i$, $a_{k+1}>i+1$,
then the integrand is $\frac{h+z_i-z_{i+1}}{w_k-z_i}$ times
a symmetric function of $z_i,z_{i+1}$. Applying $\hat s_i$
results in $-\frac{(h+w_k-z_i)(h+z_i-z_{i+1})}{(w_k-z_i)(w_k-z_{i+1})}$
times the same function, which is nothing but minus the integrand
with $a_k\to i+1$, which is precisely
$L\to s_i(L)$. That is, $\hat s_i I_L=-I_{s_i(L)}$.
\item The case $L_i=1$, $L_{i+1}=2$ is treated similarly.
\end{itemize}

Therefore all the hypotheses of Lemma~\ref{lem:deg} are satisfied,
and $\hat I_\la$ is proportional to $I_\la$.
In order to fix the normalization,
we consider the case
$a=(n-p+1,\ldots,n)$, i.e., $L(a)=L_0$. Then the integrals
can be computed one by one as follows.
The integral over $w_p$ has only one pole outside the contour,
at $z_n-h$. Next, the integral over $w_{p-1}$ has two poles,
at $z_n-h$ and $z_{n-1}-h$, but the first one is cancelled
by the factor $w_p-w_{p-1}$ in the numerator (since we have
taken the residue at $w_p=z_n-h$). So there is only one contribution,
the residue at $w_{p-1}=z_{n-1}-h$; and so on. In the end,
evaluating the residues
at $w_k=z_{k+n-p}-h$ results in:
$
\hat I_{L_0}=
\prod_{1\le i<j\le n-p} (h+z_i-z_j)
\prod_{n-p+1\le i<j\le n} (h+z_i-z_j)
$,
which coincides with $I_{L_0}$, so that $\hat I_\la=I_\la$.
\end{proof}

\def\urlshorten http://#1/#2urlend{{#1}}%
\renewcommand\url[1]{%
\href{#1}{\scriptsize\expandafter\path\urlshorten#1 urlend}%
}
\gdef\MRshorten#1 #2MRend{#1}%
\gdef\MRfirsttwo#1#2{\if#1M%
MR\else MR#1#2\fi}
\def\MRfix#1{\MRshorten\MRfirsttwo#1 MRend}
\renewcommand\MR[1]{\relax\ifhmode\unskip\spacefactor3000 \space\fi
\MRhref{\MRfix{#1}}{{\scriptsize \MRfix{#1}}}}
\renewcommand{\MRhref}[2]{%
\href{http://www.ams.org/mathscinet-getitem?mr=#1}{#2}}

\bibliography{biblio}
\bibliographystyle{amsalphahyper}

\end{document}